\documentclass[runningheads]{llncs}

\usepackage[T1]{fontenc}
\usepackage{graphicx}
\usepackage[utf8]{inputenc}
\usepackage{amsfonts,amsmath}
\usepackage{enumerate}
\usepackage{mathrsfs}
\usepackage{comment}
\usepackage{graphicx} 
\usepackage{commath}
\usepackage{afterpage}
\usepackage{xcolor}
\usepackage{listings}
\usepackage{paralist}
\usepackage{enumitem}
\usepackage{caption}
\usepackage{subcaption}
\usepackage{changepage}
\usepackage{thmtools,thm-restate}
\usepackage{mathtools}
\usepackage{multirow}
\usepackage{pbox}
\usepackage[T1]{fontenc}
\usepackage[scaled=0.85]{beramono}
\usepackage{tikz}
\usepackage{hyperref}

\usetikzlibrary{arrows}
\usetikzlibrary{shapes}
\usetikzlibrary{calc}
\usetikzlibrary{automata}
\usetikzlibrary{positioning}
\usetikzlibrary{angles}

\tikzstyle{ang}=[regular polygon, regular polygon sides = 3,draw,inner sep=0pt,minimum size=6mm, yshift = -0.75 mm]
\tikzstyle{dem}=[shape=diamond,draw,inner sep=0pt,minimum size=6mm]
\tikzstyle{ran}=[shape=circle,draw,inner sep=0pt,minimum size=6mm]
\tikzstyle{det}=[shape=rectangle,draw,inner sep=1pt,minimum size=5mm]
\tikzstyle{tran}=[draw,->,>=stealth, rounded corners]
\tikzstyle{pt}=[shape=circle,draw,inner sep=0pt,minimum size=0mm]
\tikzstyle{trannoarrow}=[draw,>=stealth, rounded corners]

\begin{document}
	
\newcommand{\PP}{{\sc PP}}
\newcommand{\support}{\mathit{supp}}
\newcommand{\dom}{\mathit{dom}}
\newcommand{\Vout}{V_{\mathit{out}}}
\newcommand{\pCFG}{\mathcal{C}}
\newcommand{\locinit}{\loc_{\mathit{in}}}
\newcommand{\locterm}{\loc_{\mathit{out}}}
\newcommand{\locfail}{\loc_{\mathit{fail}}}
\newcommand{\vecinit}{\mathbf{x}_{\mathit{in}}}
\newcommand{\vecterm}{\mathbf{x}_{\mathit{out}}}
\newcommand{\lem}{\eta}
\newcommand{\preexp}[1]{\mathit{pre}_{#1}}
\newcommand{\varinit}{\Theta_{\mathit{init}}}
\newcommand{\updates}{\mathit{Up}}
\newcommand{\Fpath}{\mathit{Fpath}}
\newcommand{\Reach}{\mathit{Reach}}
\newcommand{\Obs}{\mathit{Obs}}
\newcommand{\State}{\mathit{State}}
\newcommand{\Run}{\mathit{Run}}
\newcommand{\RunObs}{\mathit{Run}_{\texttt{observe}}}
\newcommand{\RunFail}{\mathit{Run}_{\texttt{fail}}}
\newcommand{\transitions}{\mapsto}
\newcommand{\probdist}{\mathit{Pr}}
\newcommand{\guards}{G}
\newcommand{\prob}{\mathit{Pr}}
\newcommand{\Inf}{\mathit{Inf}}
\newcommand{\Unif}{\mathit{Uniform}}
\newcommand{\Normal}{\mathit{Normal}}
\newcommand{\predicate}{\Psi}
\newcommand{\pvars}{V}
\newcommand{\rvars}{R}
\newcommand{\locs}{\mathit{L}}
\newcommand{\loc}{\ell}
\newcommand{\tran}{\tau}
\newcommand{\probm}{\mathbb{P}}
\newcommand{\E}{\mathbb{E}}
\newcommand{\expr}{E}
\newcommand{\val}{\mathbf{x}}
\newcommand{\todoP}[1]{\textbf{\textcolor{red}{TODO} P: #1}}
\newcommand{\todoDj}[1]{\textbf{\textcolor{blue}{TODO} Dj: #1}}
\newcommand{\Rset}{\mathbb{R}}
\newcommand{\sigmaAlg}{\mathcal{F}}
\newcommand{\TimeTerm}{T}
\newcommand{\Termset}{\mathit{Term}}
\newcommand{\Output}{\mathit{Output}}
\newcommand{\Restart}{\mathit{Restart}}
\newcommand{\prestart}{p_{\mathit{restart}}}
\newcommand{\TV}{\mathit{TV}}
\newcommand{\K}{\mathcal{K}}
\newcommand{\out}{\mathit{out}}
\newcommand{\Rsetnn}{\Rset_{\geq 0}}
\newcommand{\probmalt}{\mu}
\newcommand{\veca}[1]{\mathbf{#1}}
\newcommand{\nextv}{\mathit{Next}}
\newcommand{\stime}{T}
\newcommand{\indicator}[1]{\mathbb{I}_{#1}}
\newcommand{\Nset}{\mathbb{N}}
\newcommand{\Mono}{\textit{Mono}}
\newcommand{\pcfg}{\pCFG}
\newcommand{\weights}{w}
\newcommand{\iver}[1]{[#1]}
\newcommand{\pw}{W}
\newcommand{\pCFGrestarted}{\pCFG_{\text{restarted}}}
\newcommand{\pCFGaux}{\pCFG_{\text{aux}}}
\newcommand{\Prestarted}{P_{\text{res}}}
\newcommand{\Paux}{P_{\text{aux}}}

\lstdefinelanguage{affprob}
{
	morekeywords={angel,demon, choice, prob, if, then, else, fi,
		while, do, od,
		true, false, and, or, skip, sample, assume, observe, score, return, go-to},
	sensitive = false,
%	morecomment=[l]{//},
	escapeinside={@:}{:@}
}
\lstset{language=affprob}
\lstset{tabsize=3,mathescape}

\newcommand{\DZ}[1]{{\color{red}\textbf{DZ:} #1}}

\title{Refuting Equivalence in Probabilistic Programs with Conditioning}

%\title{Refuting Distributional Equivalence in~Probabilistic Programs with Conditioning}

\author{Krishnendu Chatterjee\inst{1}\orcidID{0000-0002-4561-241X} \and
Ehsan Kafshdar Goharshady\inst{1}\orcidID{0000-0002-8595-0587} \and
Petr Novotn\'y\inst{2}\orcidID{0000-0002-5026-4392} \and
\DJ{}or\dj{}e \v{Z}ikeli\'c\inst{3}\orcidID{0000-0002-4681-1699}
}

\institute{Institute of Science and Technology Austria (ISTA), Klosterneuburg, Austria \email{\{krishnendu.chatterjee, egoharsh\}@ist.ac.at} \and
Masaryk University, Brno, Czech Republic\\ \email{petr.novotny@fi.muni.cz}
\and
Singapore Management University, Singapore\\
\email{dzikelic@smu.edu.sg}}

\titlerunning{Refuting Equivalence in Probabilistic Programs with Conditioning}
\authorrunning{K.~Chatterjee, E.~K.~Goharshady, P.~Novotn\'y, \DJ.~\v{Z}ikeli\'c}

\maketitle

\begin{abstract}
We consider the problem of refuting equivalence of probabilistic programs, i.e., the problem of proving that two probabilistic programs induce different output distributions. We study this problem in the context of programs with conditioning (i.e., with observe and score statements), where the output distribution is conditioned by the event that all the observe statements along a run evaluate to true, and where the probability densities of different runs may be updated via the score statements. Building on a recent work on programs without conditioning, we present a new equivalence refutation method for programs with conditioning. Our method is based on weighted restarting, a novel transformation of probabilistic programs with conditioning to the output equivalent probabilistic programs without conditioning that we introduce in this work. Our method is the first to be both a)~fully automated, and b)~providing provably correct answers. We demonstrate the applicability of our method on a set of programs from the probabilistic inference literature.
\end{abstract}
 
\section{Introduction}
\label{sec:intro}

\emph{Probabilistic programs.} Probabilistic programs (PPs) are standard programs extended with \emph{probabilistic instructions,} such as sampling of variable values from probability distributions. PPs are used in numerous application domains, such as networking~\cite{FosterKMR016}, privacy and security~\cite{BartheGGHS16,BartheGHP16}, or planning and robotics~\cite{Ghahramani15,Thrun00}, to name a few. Unlike deterministic programs, which (assuming they terminate) for each concrete input produce a single output, PPs induce \emph{output distributions}, i.e. \emph{probability distributions} over the space of possible outputs.

\smallskip
\noindent
\emph{Conditioning.} While PPs with sampling instructions can implement various randomized algorithms, the recent proliferation of probabilistic programming is driven largely by the advances in machine learning, where PPs serve as tools for {Bayesian inference.} For this, a PP has to be equipped with \texttt{observe} statements which can be used to condition the data in program's variables through observations, as well as \texttt{score} statements which can be used to assign different probability densities to different program executions. %Formally, given an assertion \( \psi \) over program variables, the statement \texttt{observe(\( \psi \))} blocks all runs that violate \( \psi \) when passing through the statement. The statement $\texttt{score}(e)$ multiplies the probability density of all runs passing through the statement by the value of a non-negative expression $e$ over program variables.
Probability of the non-blocked runs is then renormalized~\cite{Katoen2015,DBLP:journals/toplas/OlmedoGJKKM18,BGV18}. In this way, the program's output distribution is \emph{conditioned} by the event that the program's run satisfies all observe-statement assertions it encounters and all program run probability densities are {\em renormalized} according to their cummulative score values.
%CITE "understanding of probabilistic programming by JP and also other stuff"

\smallskip
\noindent 
\emph{Static and relational analysis of PPs.} Given their inherent randomness, PPs are extremely difficult to debug. This has spurred demand for formal analyses of PPs so as to obtain provable assurances of their correctness. Over the past decade, numerous static analyses for PPs were developed, focusing mostly on properties of \emph{individual programs:} termination~(e.g.~\cite{MM05,ChakarovS13,HolgerPOPL,ChatterjeeFNH18,ChatterjeeNZ17,AgrawalC018,McIverMKK18,CheH:2020:omega-reg-decomposition,ChatterjeeGNZZ21,TakisakaOUH21,DBLP:journals/corr/abs-2304-11363,Rupak24POPL}), safety~\cite{SankaranarayananCG13,ChatterjeeNZ17,TakisakaOUH21,ChatterjeeGMZ22,BeutnerOZ22,BatzCJKKM23,WangYFLO24}, runtime and cost analysis~\cite{NgoC018,Wang0GCQS19,AvaMS:20:modular-cost,Wang0R21,DasWH:23:seesion-types,ChatterjeeGMZ24}, or input-output behavior~\cite{ChenKKW22}. Recently, attention has been increasingly given to verifying \emph{relational properties} of PPs, i.e., properties of program \emph{pairs,} such as sensitivity~\cite{BartheEGHS18,HuangWM18,WangFCDX20,0001BHKKM21} or differential privacy~\cite{BartheGGHS16,AlbarghouthiH18}. 

\smallskip
\noindent
\emph{Refuting equivalence of PPs.} In this work, we study the relational property of program \emph{equivalence}. Two PPs with fixed inputs are called equivalent if they induce the same output distributions~\cite{MurawskiO05,LegayMOW08,McIverM20}. Equivalence analysis allows one to check, e.g., whether compiler optimizations preserve the input-output behavior of the program, or if a sampler follows a desired distribution~\cite{ChakrabortyM19}. %On the other hand, the similarity problem is concerned with bounding the distance between the output distributions of two PPs. In this work, we measure the Kantorovich distance between two PPs~\cite{villani2021topics}. Kantorovich distance is an important class of distances and other distances such as Total Variation (TV) can be reduced to it. See Section~\ref{sec:problem} for a formal definition of Kantorovich distance and a discussion of its relation to other distances.
We focus on a bug-hunting aspect of the equivalence property, i.e., \emph{equivalence refutation:} given two PPs, prove that their output distributions differ. Recently, the work~\cite{chatterjee2024equivalence} presented an automated method for equivalence refutation in PPs that has formal correctness guarantees. However, the method of~\cite{chatterjee2024equivalence} works only for PPs \emph{without conditioning.} In this paper, we focus on refuting equivalence of PPs \emph{with conditioning}, i.e., proving the inequality of their \emph{conditional output distributions}.

\smallskip
\noindent
\textbf{Our contributions.} We present a new method for refuting equivalence of PPs. Our method is the first to have \emph{all} of the following properties: (i) It is fully automated, (ii) it is applicable infinite-state, discrete and continuous, Turing-complete PPs \emph{with conditioning}, and (iii) it provides provably correct refutations.

We build on the notions of \emph{upper and lower expectation supermartingales (U/LESMs)} presented in~\cite{chatterjee2024equivalence}. U/LESMs map program states to real numbers, providing upper and lower bounds on the expected value of some function \( f \) \emph{at output.} During the refutation process, the concrete function \( f \) is synthesized along with the corresponding UESM and LESM to provide the required refutation proof: for a given pair \( (P_1, P_2) \) of PPs, we seek \( f \) whose expected value at the output of \( P_1 \) is provably different (due to the bounds provided by U/LESMs) from its expected value at the output of \( P_2 \) - this witnesses the difference of the underlying output distributions. The computation of \( f \) and the U/LESMs proceeds via polynomial constraint solving. However, the method is restricted to PPs without conditioning, and its extension to PPs with conditioning was left as an open problem (see~\cite[Section~9]{chatterjee2024equivalence}). Our contributions are as follows:
\begin{compactenum}
	\item {\em Equivalence refutation for PPs with conditioning.} We extend the above to the much more general setting of PPs with conditioning, allowing for both \texttt{observe} and \texttt{score} statements. We achieve this via the notion of {\em weighted restarting in PPs} that we introduce in this work. Weighted restarting is a construction that translates a PP with conditioning into a PP without conditioning, whose output distribution is equivalent to that of the original PP. Hence, by applying weighted restarting to two PPs with conditioning, we can reduce their equivalence refutation to refuting equivalence of a pair of PPs without conditioning, for which we can utilize the method of~\cite{chatterjee2024equivalence}. %The formal reduction requires additional care though, as our weighted restarting method requires the usage of \texttt{go-to} statements in programs in order to remove conditioning instructions. In contrast, the PP syntax in~\cite{chatterjee2024equivalence} does not support \texttt{go-to} statements. To overcome this technical difficulty, we show that the method naturally extends to refuting equivalence and similarity in PPs with \texttt{go-to} statements. This allows us to combine our weighted restarting together with the prior method of~\cite{chatterjee2024equivalence}, in order to obtain a new method for equivalence and similarity refutation in PPs with conditioning.
	
	\item {\em Soundness and completeness guarantees.} While the method of~\cite{chatterjee2024equivalence} was proved to be sound, completeness of the proof rule was not studied and, to the best of our knowledge, no sound and complete proof rule for equivalence refutation in PPs has been proposed. In this work, we show that the proof rule based on U/LESMs is {\em sound and complete} for refuting equivalence. This result is novel both in the settings of PPs with and without conditioning. Moreover, it also allows us to prove that the polynomial constraint solving algorithm for computing a function $f$ over program outputs together with the U/LESM is {\em semi-complete}, in the sense that it is guaranteed to be able to refute equivalence of PPs whenever equivalence can be witnessed by a function $f$ and a pair of U/LESM that can be specified via polynomial functions of some maximal polynomial degree $D$. Note that this is a significant guarantee, given that the problem of equivalence checking of two programs is undecidable~\cite{GoldblattJ12}, hence no algorithm can be both sound and complete.
	
	%, for the special case of programs with either (1)~statically bounded termination time, or (2)~bounded variable values, the proof rule based on U/LESMs is {\em sound and complete} for refuting equivalence and similarity of PPs with conditioning. This result is novel both in the settings of PPs with conditioning and PPs without conditioning. Moreover, it also allows us to prove that the polynomial constraint solving algorithm for computing a function $f$ over program outputs together with the U/LESM is {\em semi-complete}, in the sense that it is guaranteed to be able to refute equivalence and similarity of PPs whenever equivalence and similarity can be witnessed by a function $f$ and a pair of U/LESM that can be specified via polynomial functions of some maximal polynomial degree $D$.
	
	\item {\em Experimental evaluation.} We evaluate our approach on benchmarks from the well known PSI inference tool~\cite{DBLP:conf/cav/GehrMV16}, as well as on several loopy programs. We compare with a PSI+Mathematica~\cite{Mathematica} pipeline. Our method significantly outperforms the baseline in terms of the number of refuted~equivalences.
\end{compactenum}
Finally, while the focus of this work is on equivalence refutation, weighted restarting also allows us to design a new method for {\em similarity refutation} in PPs with conditioning. The similarity refutation problem is concerned with computing a lower bound on the distance between output distributions of two PPs. The work~\cite{chatterjee2024equivalence} proposed a method for similarity refutation in PPs without conditioning, and in this work we generalize it to the much more general setting of PPs with conditioning. Due to space restrictions, we defer this result to Appendix~\ref{app:similarity}. However, in our experiments in Section~\ref{sec:experiments}, for each refuted equivalence we also report a lower bound on the Kantorovich distance~\cite{villani2021topics} computed by our method.

\smallskip
\noindent
\emph{Related work.} The study of PP equivalence is partly motivated by the necessity to reason about the correctness and accuracy of probabilistic inference algorithms. Previous approaches to this issue~\cite{BatuFRSW13,ChanDVV14,GrosseGA15,GrosseAR16,Cusumano-Towner17,canonne2020survey,Domke2021} employ sampling methods that provide statistical guarantees on the correctness of the result. %In contrast, our method provides provably correct equivalence refutations.

Relational analyses of PPs mostly focused on \emph{sensitivity properties}~\cite{BartheEGHS18,HuangWM18,WangFCDX20,0001BHKKM21}, often using \emph{coupling proofs}~\cite{BartheGGHS16,AlbarghouthiH18}. Intuitively, given a probabilistic program and its two inputs, the goal here is to bound the distance of output distributions corresponding to these inputs. In particular, the relational analysis in this context focuses on relations between different \emph{inputs,} rather than different \emph{programs}. Indeed, the previous approaches to sensitivity analysis typically exploit the "aligned" control flow of two executions on sufficiently close inputs (with a notable exception of the recent work~\cite{Gregersen2024asynchronous} which aims to dispel the alignment requirement.) Our method makes no such assumptions. Moreover, none of the aforementioned works except for~\cite{HuangWM18} considers programs with conditioning.

%As discussed in~\cite{chatterjee2024equivalence}, the concept of U/LESMs is inspired by the notions of upper and lower \emph{cost supermartingales}~\cite{Wang0GCQS19} and super- and sub-invariants for expectation bounds~\cite{HarkKGK20}. The key differences of U/LESMs from the previous approaches are summarized in Section~5.2. of~\cite{chatterjee2024equivalence}. In a nutshell, none of these previous approaches has been applied to relational analysis, and the use of U/LESMs for such analysis requires innovative steps, such as employing both upper- and lower bounds within a single rule and synthesizing the function \( f \) (whose expectation is being bounded by the equivalence refutation rule) \emph{together} with the U/LESMs that provide the expectation bounds. Moreover, the U/LESM-based approach works for programs with both positive- and negative-valued variables, whereas the assumption of non-negative variables is imposed both in~\cite{HarkKGK20} and in other automated approaches to computing expectation bounds~\cite{AvaMS:20:modular-cost,Wang0R21}. Also, the aforementioned works \emph{do not} consider programs with conditioning.

The work~\cite{ZikelicCBR22} uses ``potential functions'' to analyze resource usage difference of a program pair. Potential functions for the two programs are computed simultaneously, similarly to our approach. However, the work does not involve probabilistic programs and its overall aim (cost analysis) differs from ours.

\section{Preliminaries}\label{sec:prelims}

%We use boldface notation to represent vectors, e.g.~$\mathbf{x}$ or $\mathbf{y}$. We use $\mathbf{x}[i]$ to denote the $i$-th component of $\mathbf{x}$. For a scalar $a$, we use $\mathbf{x}(i \leftarrow a)$ to denote the vector $\mathbf{y}$ such that $\mathbf{y}[i] = a$ and $\mathbf{y}[j] = \mathbf{x}[j]$ for $j \neq i$.
We assume that the reader is familiar with notions from probability theory such as probability measure, random variable and expected value. The term \emph{probability distribution} refers to a probability measure over a subset of Euclidean space $\mathbb{R}^n$. We use boldface notation to represent vectors, e.g.~$\mathbf{x}$ or $\mathbf{y}$.

\smallskip\noindent{\em Syntax.} We consider imperative arithmetic probabilistic programs (PPs) allowing standard programming constructs such as variable assignments, if-branching, sequential composition, go-to statements, and loops. All program variables are assumed to be integer- or real-valued. For PP semantics to be well defined, we assume that all arithmetic expressions are Borel-measurable functions~\cite{Williams91}.

In addition, we allow two types of probabilistic instructions, namely {\em sampling} and {\em conditioning}. Sampling instructions appear on the right-hand-side of program variable assignments and are represented in our syntax by a command $\textbf{sample}(d)$, where $d$ is a probability distribution. We do not impose specific restrictions and allow sampling instructions from both discrete and continuous distributions. Importantly, we allow the two standard conditioning instructions: $\textbf{observe}(\phi)$, where $\phi$ is a boolean predicate over program variables, and \( \textbf{score}(e) \), where \( e \) is a non-negative expression over program variables. A program is \emph{conditioning-free} if it does not contain any conditioning instructions.

Probabilistic branching instructions, i.e.~$\textbf{if prob}(p)$, can be obtained as syntactic sugar by using a sampling instruction followed by conditional branching. %Hence, we will use probabilistic branching in our examples, while omitting it from the formal description in the interest of space.

\smallskip\noindent{\em Probabilistic control flow graphs (pCFGs).} In order to formalize the semantics of our PPs and to present our results and algorithms, we consider an operational representation of PPs called {\em probabilistic control-flow graphs (pCFGs)} extended with the concept of {\em weighting} as in~\cite{WangYFLO24}. The use of pCFGs in the analysis of PPs is standard, see e.g.~\cite{ChatterjeeFNH18,AgrawalC018}. Hence, we keep this exposition brief. Formally, a pCFG is a tuple $\pCFG=(\locs,V,\Vout,\locinit,\vecinit,\transitions,\guards,\updates,\weights,\locterm)$, where:
\begin{compactitem}
	\item $\locs$ is a finite set of {\em locations};
	\item $V=\{x_1,\dots,x_{|V|}\}$ is a finite set of {\em program variables};
	\item $\Vout = \{x_1,\dots,x_{|\Vout|}\} \subseteq V$ is a finite set of {\em output variables};
	\item $\locinit \in \locs$ is the {\em initial location} and $\vecinit \in \mathbb{R}^{|V|}$ the {\em initial variable valuation};
	\item $\transitions\,\subseteq \locs \times \locs$ is a finite set of {\em transitions}; %of the form $\tau=(\loc,\loc')$, where $\loc$ is the {\em source location} and $\loc'$ is the {\em successor locations} of $\tau$;
	%We require that $\sum_{\loc' \in \locs}\prob(\loc') = 1$;
	\item $\guards$ is a map assigning to each transition $\tau\in\,\transitions$ a {\em guard} $\guards(\tau)$, which is a predicate over $V$ specifying whether $\tau$ can be executed.
	\item $\updates$ is a map assigning to each transition $\tau\in\,\transitions$ an {\em update} $\updates(\tau)=(u_1,\dots,u_{|V|})$ where for each $j\in\{1,\dots,|V|\}$ the variable update $u_j$ is either
	\begin{compactitem}
		\item the bottom element $\bot$ denoting no variable update, or
		\item a Borel-measurable arithmetic expression $u:\mathbb{R}^{|V|}\rightarrow\mathbb{R}$, or
		\item a probability distribution $u = \delta$;
	\end{compactitem}
	\item \( \weights\colon \transitions \times \Rset^{|V|} \rightarrow[0,\infty)\) is a {\em weighting function,} and
	\item $\locterm \in \locs$ the {\em terminal location}.
\end{compactitem}
Translation of PPs into pCFGs is standard~\cite{ChatterjeeFNH18,AgrawalC018,WangYFLO24}, hence we omit the details. We denote by \( \pCFG(P) \) the pCFG induced by a program \( P \).
Intuitively, locations of \( \pCFG(P) \) correspond to program instructions, while transition(s) out of a location encodes the effects of the respective instruction. Weight of a transition is~1 unless the transition encodes one of the conditioning instructions: If \( \tau \) is encodes a \( \textbf{score}(e) \) instruction, we put \( \weights(\tau,\mathbf{x}) = e(\mathbf{x}) \). %If \( \tau \) does not involve any conditioning, we put \( \weights(\tau,\mathbf{x}) = 1 \) for all \( \vec{x} \).
If \( \tau \) encodes an \( \textbf{observe}(\phi) \) statement, we put
\( \weights(\tau,\mathbf{x}) = 1 \) if \( \mathbf{x} \models \phi \) and \( \weights(\tau,\mathbf{x}) = 0 \) otherwise. Note that if \( P \) is {conditioning-free} then in \( \pCFG(P) \) we have  \( \weights(\tau,\mathbf{x}) = 1 \) for all  \( \tau \) and \( \vec{x} \).

%To encode observation instructions into pCFGs, we introduce one location $\loc$ for each observation statement $\textbf{observe}(\phi)$ in the program. The location $\loc$ has two outgoing transitions -- transition $\tau = (\loc,\loc')$ to the location $\loc'$ corresponding to the next program instruction with guard $\guards(\tau) = \phi$ and with no variable update, and transition $\tau_{\text{fail}} = (\loc,\locfail)$ to the observation violation location with guard $\guards(\tau_{\text{fail}}) = \neg \phi$ and with update which resets all variables to $0$. The observation violation location $\locfail$ has a single outgoing transition $\tau = (\locfail,\locterm)$ to $\locterm$ with guard $\guards(\tau) = \mathit{true}$ and with no variable update. Finally, the terminal location $\locterm$ has a single outgoing transition which is a self-loop with no variable update.

%W.l.o.g. we assume that each location $\loc$ has at least one outgoing transition and that the pCFG has no dead-ends, i.e. the disjunction of guards of all transitions outgoing from $\loc$ is equivalent to $\mathit{true}$. To ensure that there is no non-determinism in our PPs, we also require that guards of two transitions $\tau_1$ and $\tau_2$ outgoing from $\loc$ are {\em mutually exclusive}, i.e.~$\guards(\tau_1)\land\guards(\tau_2)\equiv\mathit{false}$. These assumptions are satisfied by default by all pCFGs induced by PPs in our syntax.

\begin{example}[Running example]
	Fig.~\ref{fig:running} shows a PP pair that will serve as our running example. Program variables in both PPs are $V = \{x,y,c,r\}$, with output variables $\Vout = \{x,y\}$. In each loop iteration, with probability $0.5$ the variable $x$ is decremented by a random value sampled according to a beta distribution, whereas with probability $0.5$ the variable $y$ is incremented by $1$. Hence, $y$ tracks the number of loop iterations in which $x$ is not decremented. Upon loop termination, the \textbf{score} command is used to multiply the probability density of each run by the value of $y$ upon termination. Thus, larger probability density is given to program runs that modify $x$ a smaller number of times. The programs differ only in the parameters of the beta distribution (highlighted in blue).
\end{example}

\begin{figure}[t]
	\begin{minipage}[c]{0.45\textwidth}
		\begin{lstlisting}[mathescape]
	$\locinit$:   $x,y,c,r = 0$
	$\loc_1$:  while $c \leq 999$:
	$\loc_2$:		if prob($0.5$):
	$\loc_3$:			$\textcolor{blue}{\textit{r} := \textbf{Beta}(2,2)}$
	$\loc_4$:			$x:= x - r$
			else:
	$\loc_5$:			$y := y + 1$
	$\loc_6$:		$c = c + 1$
	$\loc_7$:$\,\,\,\,$   score($y$)
	$\locterm$: return($x,y$)
		\end{lstlisting}
	\end{minipage}
	\begin{minipage}[c]{0.45\textwidth}
		\begin{lstlisting}[mathescape]
		$\locinit$: $x,y,c,r = 0$
		$\loc_1$:  while $c \leq 999$:
		$\loc_2$:		if prob($0.5$):
		$\loc_3$:			$\textcolor{blue}{\textit{r} := \textbf{Beta}(2,3)}$
		$\loc_4$:			$x:= x - r$
				else:
		$\loc_5$:			$y := y + 1$
		$\loc_6$:		$c = c + 1$
		$\loc_7$:$\,\,\,\,$   score($y$)
		$\locterm$: return($x,y$)
		\end{lstlisting}
	\end{minipage}
	\vspace{-1em}
	\caption{Running example. Labels $\locinit$, $\loc_1, \dots, \loc_7, \locterm$ correspond to locations in the probabilistic control-flow graphs (pCFGs) of both programs.}
	\label{fig:running}
	\vspace{-1em}
\end{figure}

\noindent{\em States, paths and runs.} A {\em state} in a pCFG $\pCFG$ is a tuple $(\loc,\mathbf{x})$, where $\loc$ is a location and $\mathbf{x}\in\mathbb{R}^{|V|}$ is a variable valuation. A transition $\tau=(\loc,\loc')$ is {\em enabled} at a state $(\loc,\mathbf{x})$ if $\mathbf{x}\models\guards(\tau)$. A state $(\loc',\mathbf{x}')$ is a {\em successor} of $(\loc,\mathbf{x})$, if there exists an enabled transition $\tau=(\loc,\loc')$ in $\pCFG$ such that $\mathbf{x}'$ is a possible result of applying the update of $\tau$ to $\mathbf{x}$. The state $(\locinit,\vecinit)$ is the {\em initial state}. A state $(\loc,\mathbf{x})$ is said to be {\em terminal}, if $\loc=\locterm$. We use $\State^\pCFG$ to denote the set of all states~in~$\pCFG$.

A {\em finite path} in $\pCFG$ is a sequence of successor states $(\loc_0,\mathbf{x}_0),(\loc_1,\mathbf{x}_1),\dots,(\loc_k,\mathbf{x}_k)$ with $(\loc_0,\mathbf{x}_0)=(\locinit,\vecinit)$. A state $(\loc,\mathbf{x})$ is {\em reachable} in $\pCFG$ if there exists a finite path in $\pCFG$ whose last state is $(\loc,\mathbf{x})$. A {\em run} (or an {\em execution}) in $\pCFG$ is an infinite sequence of states whose each finite prefix is a finite path. 
%A run is said to be {\em observation non-violating}, if it does not contain any state with location $\locfail$. 
We use $\Fpath^\pCFG$, $\Run^\pCFG$, and $\Reach^\pCFG$ to denote the sets of all finite paths, runs, and reachable states in~$\pCFG$.

A run $\rho \in \Run^\pCFG$ is {\em terminating} if it reaches some terminal state $(\locterm,\mathbf{x})$. We use $\Termset\subseteq\Run^{\pCFG}$ to denote the set of terminating runs in $\Run^{\pCFG}$. We define the {\em termination time} of $\rho$ via $\TimeTerm(\rho) = \inf_{i \geq 0}\{i \mid \loc_i = \locterm \}$, with $\TimeTerm(\rho) = \infty$ if $\rho$ is not terminating. A \emph{cumulative weight} of a terminating run \( \rho = (\loc_0,\mathbf{x}_0),(\loc_1,\mathbf{x}_1),\dots\) is 
\( \pw(\rho) = \prod_{i=0}^{\TimeTerm(\rho)-1} \weights((\loc_i,\loc_{i+1}),\mathbf{x}_i), \) with an empty product being equal to 1. 

\smallskip\noindent{\em Semantics.} The pCFG semantics are formalized as Markov chains with weights~\cite{WangYFLO24}. In particular, a pCFG $\pCFG$ defines a discrete-time Markov process over the states of $\pCFG$ whose trajectories correspond to runs in $\pCFG$. Each trajectory starts in the initial state $(\locinit,\vecinit)$ with weight $1$. Then, at each time step $i \in \mathbb{N}_0$, if the trajectory is in state $(\loc_i,\mathbf{x}_i)$, the next trajectory state $(\loc_{i+1},\mathbf{x}_{i+1})$ is defined according to the unique  pCFG transition $\tau_i$ enabled at $(\loc_i,\mathbf{x}_i)$, and the weight of the trajectory is multiplied by the value of the weight function $w(\tau_i,\mathbf{x}_i)$. 
%Formally, we set $\mathbf{x}_{i+1}[j] = \nextv(\tau_i,\val_i)[j]$, where \( \nextv(\tau_i,\val_i) \) is the \emph{next value} vector defined for each \( 1 \leq j \leq n \) according to the update element \( u_j = \updates(\tau_i)[j].\) See Appendix~\ref{xxx} for details.
% \): i) if \( u_j = \bot\) then \( \nextv(\tau_i,\val_i)[j] = \val_i[j] \); ii) if \( u_j \) is an arithmetic expression, then \( \nextv(\tau_i,\val_i)[j] = u_j(\val_i) \); and iii) if \( u_j \) is a probability distribution, then \( \nextv(\tau_i,\val_i)[j] \) is sampled from \( u_j \). 

This Markov process gives rise to a probability space $(\Run^\pCFG, \mathcal{F}^\pCFG, \probm^\pCFG)$ over the set of all pCFG runs, formally defined via the cylinder construction~\cite{MeynTweedie}. We use $\mathbb{E}^\pCFG$ to denote the expectation operator in this probability space.

\smallskip\noindent{\em Almost-sure termination.} We restrict our attention to almost-surely terminating programs, which is necessary for them to define valid probability distributions on output. A program $P$ (or, equivalently, its pCFG \( \pCFG(P) \)) terminates {\em almost-surely (a.s.)} if $\mathbb{P}^{\pCFG(P)}[\Termset] = 1$. A.s.~termination can be automatically verified in arithmetic PPs, by e.g.~synthesizing a ranking supermartingale (RSM)~\cite{ChakarovS13,ChatterjeeFG16}.

\section{Equivalence Refutation Problem}\label{sec:problem}

%We start by recalling the notion of almost-sure termination of PPs, which is a necessary assumption for PPs to define valid probability distributions on output.

%\smallskip\noindent{\em Almost-sure termination.} A pCFG $\pCFG$ terminates {\em almost-surely (a.s.)} if $\mathbb{P}^\pCFG[\Termset] = 1$. Almost-sure termination can be checked in a fully automated fashion by e.g.~synthesizing a ranking supermartingale (RSM)~\cite{ChakarovS13,ChatterjeeFG16}.

For each variable valuation $\mathbf{x} \in \mathbb{R}^{|V|}$, we use $\mathbf{x}^{\out}$ to denote its projection to the components of output variables $\Vout$. For a terminating run $\rho$ that reaches a terminal state $(\locterm,\mathbf{x})$, we say that $\mathbf{x}^{\out}$ is the {\em output} of \( \rho \). In an a.s. terminating pCFG \( \pCFG \), the probability measure \( \probm^\pCFG  \) over runs naturally extends to a probability measure over sets of outputs, weighted by the cumulative weights of individual runs and normalized by the total expected weight over all runs. Formally, the {\em normalized output distribution (NOD)} of $\pCFG$ is the probability distribution over all output variable valuations $\mathbb{R}^{|\Vout|}$ defined by 
\[ \mu^{\pCFG}[B] = \frac{\mathbb{E}^{\pCFG}\big[\indicator{\Output(B)}\cdot W \big]}{\mathbb{E}^{\pCFG}\big[ W\big]}, \]
where (i) $B$ is any Borel-measurable subset of \( \Rset^{|\Vout|} \), (ii) \( \indicator{A} \) is the indicator variable of a Borel-measurable set of runs \( A \), and (iii):
\[ \Output(B) = \big\{ \rho \in \Run^\pCFG \mid \rho \textit{ reaches a terminal state } (\locterm,\mathbf{x}) \text{ s.t. } \mathbf{x}^{\out} \in B \big\}. \]
For the NOD \( \mu^{\pCFG} \) to be well-defined, we restrict our attention to pCFGs where \(  0 < \mathbb{E}^{\pCFG}\big[ W] < \infty \), in which case we call \( \pCFG \) \emph{integrable} (in line with~\cite{WangYFLO24}).

\smallskip\noindent{\em Problem assumptions.} For a pair of PPs $P^1$ and $P^2$, we assume that: (1)~both $\pCFG(P^1)$ and $\pCFG(P^2)$ are a.s.~terminating and integrable, for their output distributions to be well defined, and (2)~$\pCFG(P^1)$ and $\pCFG(P^2)$ share the same set of output variables $\Vout$, for their output distributions to be defined over the same sample space so that their equivalence is well defined. Our algorithm will automatically check (1), and (2) is trivially checked by comparing the two input pCFGs.

\smallskip\noindent {\bf Problem.} Let $P^1$ and $P^2$ be two PPs with pCFGs $\pCFG(P^1)$ and $\pCFG(P^2)$ being a.s.~terminating, integrable, and having the same set of output variables $\Vout$. Our goal is to prove that $P^1$ and $P^2$ are {\em not (output) equivalent}, i.e.~that there exists a Borel-measurable set $B \subseteq \mathbb{R}^{|\Vout|}$ such that $\mu^{\pCFG(P^1)}[B] \neq \mu^{\pCFG(P^2)}[B]$.

\section{Equivalence Refutation Proof Rule}\label{sec:proofrules}

In this section, we first introduce the weighted restarting transformation, and then use it to formulate a proof rule for equivalence refutation in programs with conditioning. We also show the completeness of this proof rule.

\subsection{Weighted Restarting}\label{sec:restarting}

We introduce weighted restarting, which transforms a PP into output equivalent conditioning-free PP. Our method of weighted restarting requires PPs and their induced pCFGs to satisfy the bounded cumulative weight property.

\smallskip\noindent{\em Assumption: Bounded cumulative weight property.} We say that a program \( P \) (or equivalently, its pCFG \( \pCFG(P) \)) satisfies the {\em bounded cumulative weight property}, if there exists $M > 0$ such that the cumulative weight of every terminating run $\rho$ is bounded from above by $M$, i.e.~$W(\rho) \leq M$. Note that, in an a.s.~terminating pCFG $\pCFG$,  the bounded cummulative weight property also implies the upper bound in the integrability property, since if $W(\rho) \leq M$ holds for all terminating runs then we also have $\mathbb{E}^{\pCFG(P)}[W] \leq M < \infty$. Our algorithm for equivalence refutation in Section~\ref{sec:algorithm} will formally verify this property and compute the value of~$M$. %Formal verification of this property can be fully automated as follows. First, introduce a fresh program variable $W$ in the pCFG. The value of $W$ is initially set to $W=1$ and is updated at each transition by multiplying its value by the value of the weighting function of that transition. Second, generate a program invariant to prove the assertion $W \leq M$ at terminal location $\locterm$, where $M$ is some large constant value. Program invariant generation can be fully automated using existing invariant generation tools, e.g.~\cite{FeautrierG10,SankaranarayananSM04,Chatterjee0GG20}.

\smallskip\noindent{\em Weighted restarting.} Consider a PP $P$ which is a.s.~terminating, integrable, and satisfies the bounded cummulative weight property with bound $M>0$. Weighted restarting produces an output equivalent conditioning-free PP $\Prestarted$.
To achieve this, it introduces a fresh program variable $W$, which we call the {\em weight variable}, to capture all information encoded via observe and score instructions. This allows us to remove all conditioning instructions from the~PP. Then, for the output distribution to be equivalent to the {\em normalized output distribution} of the original PP, it adds a block of code that {\em terminates} a program run with probability proportional to its cumulative weight, and otherwise {\em restarts} it and moves it back to the initial program state. This construction is formalized as follows:
\begin{compactenum}
	\item {\em Introduce weight variable $W$.} Introduce a fresh program variable $W$. The value of $W$ is initially set to $W=1$.
	\item {\em Remove observe instructions.} Each $\textbf{observe}(\phi)$ instruction is replaced by the following block of code, which sets the cumulative weight of a run to~$0$ if the predicate in the observe instruction is false:
	\begin{equation*}
	\begin{split}
		&\textbf{if } \neg\phi:\\
		&\quad\quad W := 0
	\end{split}
	\end{equation*}
	\item {\em Remove score instructions.} Each $\textbf{score}(e)$ instruction is replaced by the following block of code, which multiplies the weight variable $W$ by the value of the expression in the score instruction:
	\begin{equation*}
		W := W \cdot e
	\end{equation*}
	\item {\em Add restarting upon termination.} Finally, the following block of code is added at the end of the PP, which with probability $1 - W/M$ ''restarts'' a program run by moving it back to the initial state and resetting its weight to $W=1$. We write "$\mathbf{x} := \vecinit$" as a shorthand for the block of code that restarts each program variable value to that specified by the initial variable valuation $\vecinit$:
	\begin{equation*}
		\begin{split}
			&\textbf{if prob}(1 - W/M): \\
			&\quad\quad \mathbf{x} := \vecinit \\
			&\quad\quad W := 1 \\
			&\quad\quad \textbf{go-to } \locinit 
		\end{split}
	\end{equation*}
\end{compactenum}
The work~\cite{OlmedoGJKKM18} also proposed a "restarting" procedure for PPs with observe, but without score, instructions. This is achieved by introducing a new boolean variable $\texttt{unblocked}$ whose value is initially true, but is set to false if any $\textbf{observe}(\phi)$ condition is violated. The program is then embedded into a "$\textbf{while } \texttt{unblocked}$:" loop. However, it is not clear how to extend this translation to PPs with score instructions, whereas our method supports both observe and score instructions.

%Note that the weighted restarting procedure does not change the set of output variables of the PP. The following theorem proves its correctness, i.e.~that it indeed gives rise to an output equivalent conditioning-free PP.

\begin{theorem}[Correctness of weighted restarting, proof in Appendix]\label{thm:weightedrestarting}
	Consider a PP $P$ which is a.s.~terminating, integrable, and satisfies the bounded cumulative weight property with bound $M>0$. Let $\Prestarted$ be a conditioning-free PP produced by the weighted restarting procedure. Then, $\Prestarted$ is also a.s.~ terminating, integrable, and satisfies the bounded cumulative weight property. Moreover, $\pCFG(P)$ and $\pCFG(\Prestarted)$ are output equivalent, i.e. \( \mu^{\pcfg(P)} = \mu^{\pcfg(\Prestarted)} \).
\end{theorem}

\begin{example}
	Fig.~\ref{fig:restarting} shows the PPs produced by the weighted restarting procedure applied to PPs in Fig.~\ref{fig:running}. Notice that the PPs in Fig.~\ref{fig:running} are a.s.~terminating, integrable, and satisfy the bounded cumulative weight property with $M=1000$, as $y \leq 1000$ holds upon loop termination. Weighted restarting produces the PPs in Fig.~\ref{fig:restarting}. %as follows. In the first step, the weight variable $W$ is introduced and is initialized to $1$. Since there are no \textbf{observe} instructions in the PPs, the second step results in no changes to the programs. In the third step, the \textbf{score} instructions in line~$6$ are replaced by the updates to the weight variable $W$. Finally,
	In the fourth step of the procedure, the restarting upon termination block of code is added to each program in lines~$7-10$. It can be realized using a single pCFG transition, hence we only need one pCFG location~$\loc_8$ for it.
\end{example}

\begin{figure}[t]
	\begin{minipage}[c]{0.45\textwidth}
		\begin{lstlisting}[mathescape]
	$\locinit$:  $x,y,c,r = 0, \textcolor{red}{W=0}$
	$\loc_1$:  while $c \leq 999$:
	$\loc_2$:		if prob($0.5$):
	$\loc_3$:			$\textcolor{blue}{\textit{r} := \textbf{Beta}(2,2)}$
	$\loc_4$:			$x:= x - r$
			else:
	$\loc_5$:			$y := y + 1$
	$\loc_6$:		$c = c + 1$
	$\loc_7$:   $\textcolor{red}{W := W \cdot y}$
	$\loc_8$:   $\textcolor{red}{\textbf{if prob}(1-W/1000):}$
			$\textcolor{red}{x,y,c,r := 0}$
			$\textcolor{red}{W := 1}$
			$\textcolor{red}{\textbf{go-to } \locinit}$
	$\locterm$: return($x,y$)
		\end{lstlisting}
	\end{minipage}
	\begin{minipage}[c]{0.45\textwidth}
		\begin{lstlisting}[mathescape]
		$\locinit$: $x,y,c,r = 0, \textcolor{red}{W=0}$
		$\loc_1$:  while $c \leq 999$:
		$\loc_2$:		if prob($0.5$):
		$\loc_3$:			$\textcolor{blue}{\textit{r} := \textbf{Beta}(2,3)}$
		$\loc_4$:			$x:= x - r$
				else:
		$\loc_5$:			$y := y + 1$
		$\loc_6$:		$c = c + 1$
		$\loc_7$:   $\textcolor{red}{W := W \cdot y}$
		$\loc_8$:   $\textcolor{red}{\textbf{if prob}(1-W/1000):}$
				$\textcolor{red}{x,y,c,r := 0}$
				$\textcolor{red}{W := 1}$
				$\textcolor{red}{\textbf{go-to } \locinit}$
		$\locterm$: return($x,y$)
		\end{lstlisting}
	\end{minipage}
	\vspace{-1em}
	\caption{PPs produced by weighted restarting applied to the PPs in Fig.~\ref{fig:running}. The changes introduced by the weighted restarting procedure are highlighted in red.}
	\label{fig:restarting}
\end{figure}

\subsection{Equivalence Refutation in PPs with Conditioning}
\label{sec:esmrecap}

Theorem~\ref{thm:weightedrestarting} effectively reduces the task of equivalence refutation for a PP pair \( (P^1, P^2) \) into the respective task for \( (\Prestarted^1, \Prestarted^2) \) of PPs without conditioning. Following such a reduction, we can apply the refutation rule from~\cite{chatterjee2024equivalence}. We briefly revisit this proof rules. We additionally show its completeness in settings both with and without conditioning, a result not proven previously.

In what follows, let  $\pCFG=(\locs,V,\Vout,\locinit,\vecinit,\transitions,\guards,\updates,\weights,\locterm)$ be an a.s.~terminating pCFG. We use the following terminology:
\begin{compactitem}
	\item A {\em state function} $\lem$ in $\pCFG$ is a function which to each location $\loc \in \locs$ assigns a Borel-measurable function $\lem(\loc):\mathbb{R}^{|V|} \rightarrow \mathbb{R}$ over program variables.
	\item A {\em predicate function} $\Pi$ in $\pCFG$ is a function which to each location $\loc \in \locs$ assigns a predicate $\Pi(\loc)$ over program variables. It naturally induces a set of states $\{(\loc,\mathbf{x}) \mid \mathbf{x} \models \Pi(\loc)\}$, which we also refer to via $\Pi$. \( \Pi \) is an {\em invariant} if it contains all reachable states in $\pCFG$, i.e.~if $\mathbf{x} \models I(\loc)$ for each $(\loc,\mathbf{x}) \in \Reach^{\pCFG}$.
\end{compactitem}

%\subsection{Recap: Expectation Super/Submartingales}

%We start by recalling the notions of UESMs and LESMs of~\cite{chatterjee2024equivalence} in an a.s.~terminating pCFG $\pCFG$. Since the definitions of UESMs and LESMs consider pCFGs rather than PPs, these notions remain well defined even in the setting of pCFGs induced by PPs with conditioning instructions.

\smallskip
\noindent
Let $f: \mathbb{R}^{|\Vout|} \rightarrow \mathbb{R}$ be a Borel-measurable function over the pCFG outputs. A UESM (resp.~LESM) for $f$ is a state function $U_f$ (resp.~$L_f$) that satisfies certain conditions in every reachable state which together ensure that it evaluates to an upper (resp.~lower) bound on the expected value of $f$ on the pCFG output. Since it is generally not feasible to exactly compute the set of reachable states, UESMs and LESMs are defined with respect to a supporting invariant that over-approximates the set of all reachable states in the pCFG. This is done in order to later allow for a fully automated computation of UESMs and LESMs in PPs. 

\begin{definition}[Upper expectation supermartingale (UESM)~\cite{chatterjee2024equivalence}]\label{def:uesm}
	Let $\pCFG$ be an a.s.~terminating pCFG, $I$ be an invariant in $\pCFG$ and $f:\mathbb{R}^{|\Vout|} \rightarrow \mathbb{R}$ be a Borel-measurable function over the output variables of $\pCFG$. An {\em upper expectation supermartingale (UESM) for $f$} with respect to the invariant $I$ is a state function $U_f$ satisfying the following two conditions:
	\begin{compactenum}
		\item {\em Zero on output.} For every $\mathbf{x} \models I(\locterm)$, we have $U_f(\locterm,\mathbf{x}) = 0$.
		\item {\em Expected $f$-decrease.} For every location $\loc\in\locs$ , transition $\tau = (\loc,\loc') \in\,\transitions$, and variable valuation \( \val \) s.t. $\val \models I(\loc) \land \guards(\tau)$, we require the following: denoting by \( \veca{N} \) the expected valuation vector after performing transition \( \tau \) from state \( (\loc,\vec{x}) \), it holds
		\begin{equation}
			\label{eq:ufexp}
			U_f(\loc,\mathbf{x}) + f(\mathbf{x}^{\out})\geq \E\Big[U_f(\ell', \veca{N}) + f(\veca{N}^\out) \Big]
		\end{equation}
		with \( \veca{N}^\out \) the projection of the random vector \( \veca{N} \) onto \( V_\out \)-indexed variables.
	\end{compactenum}
\end{definition}

\noindent Intuitively, an UESM is required to be equal to $0$ on output and, in every computation step, any increase in the \( f \)-value is {\em exceeded in expectation} by the decrease in the UESM value. A \emph{lower expectation submartingale} (LESM) \( L_f \) for \( f \) is defined analogously, with \eqref{eq:ufexp} replaced by the dual \emph{expected \( f \)-increase} condition:

%\begin{definition}[Lower expectation submartingale (LESM)~\cite{chatterjee2024equivalence}]\label{def:lesm}
%	Let $\pCFG$ be an a.s.~terminating pCFG, $I$ be an invariant in $\pCFG$ and $f:\mathbb{R}^{|\Vout|} \rightarrow \mathbb{R}$ be a Borel-measurable function over the output variables of $\pCFG$. A {\em lower expectation submartingale (LESM) for $f$} with respect to the invariant $I$ is a state function $L_f$ satisfying the following two conditions:
%	\begin{compactenum}
%		\item {\em Zero on output.} For every $\mathbf{x} \models I(\locterm)$, we have $L_f(\locterm,\mathbf{x}) = 0$.
%		\item {\em Expected $f$-increase.} For every location $\loc\in\locs$ , transition $\tau = (\loc,\loc') \in\,\transitions$, and variable valuation \( \val \) s.t. $\val \models I(\loc) \land \guards(\tau)$, we require the following: for \( \veca{N} = \nextv(\tau,\val)\) it holds
		\begin{equation}
			\label{eq:lfexp}
			L_f(\loc,\mathbf{x}) + f(\mathbf{x}^{\out})\leq \E\Big[L_f(\ell', \veca{N}) + f(\veca{N}^\out) \Big].
		\end{equation}
%	\end{compactenum}
%\end{definition}

\begin{example}\label{ex:uesm}
	Consider the two PPs obtained by weighted restarting in Fig.~\ref{fig:restarting}. Their output variables are $\Vout = \{x,y\}$. By Theorem~\ref{thm:weightedrestarting}, both PPs are a.s.~terminating. For a function over outputs $f(x,y) = x+y$, one can check by inspection that the state function $U_f$ below defines a UESM for $f$ in the first program, and the state function $L_f$ below defines an LESM for $f$ in the second program. These are also the U/LESMs produced by our prototype implementation:
	\begin{equation*}
	\resizebox{\hsize}{!}{
		$U_f \begin{pmatrix}
			x,\\
			y,\\
			c,\\
			r,\\
			W
		\end{pmatrix}
		= \begin{cases}
			500 - 0.5 \cdot y, &\text{if } \loc \in \{\locinit, \loc_1, \loc_2, \loc_6, \loc_7, \loc_8 \} \\
			499.5 - 0.5 \cdot y, &\text{if } \loc = \loc_3 \\
			500 - 0.5 \cdot y - r, &\text{if } \loc = \loc_4 \\
			500.5 - 0.5 \cdot y, &\text{if } \loc = \loc_5 \\
			0, &\text{if } \loc = \locterm
		\end{cases}
		\quad
		L_f\begin{pmatrix}
			x,\\
			y,\\
			c,\\
			r,\\
			W
		\end{pmatrix}
		= \begin{cases}
			600 - 0.6 \cdot y, &\text{if } \loc \in \{\locinit, \loc_1, \loc_2, \loc_6, \loc_7, \loc_8 \} \\
			599.6 - 0.6 \cdot y, &\text{if } \loc = \loc_3 \\
			600 - 0.6 \cdot y - r, &\text{if } \loc = \loc_4 \\
			600.4 - 0.6 \cdot y, &\text{if } \loc = \loc_5 \\
			0, &\text{if } \loc = \locterm
		\end{cases}$
	}
	\end{equation*}
\end{example}

\smallskip U/LESMs can be used to refute equivalence in conditioning-free PPs under so-called {\em OST-soundness conditions.} The name is due to the fact that correctness of the proof rule in~\cite{chatterjee2024equivalence} is proven via Optional Stopping Theorem (OST)~\cite{Williams91} from martingale theory. The first three conditions below are imposed by the classical OST, and the fourth condition is derived from the Extended OST~\cite{Wang0GCQS19}.

\begin{definition}[OST-soundness]\label{def:ostsound}
	Let \( \pCFG \) be a pCFG, \( \eta \) be a state function in \( \pCFG \), and \( f \colon \Rset^{|V_\out|} \rightarrow \Rset\) be a Borel measurable function. Let \( S_i \) be the \( i \)-th state and \( \veca{x}_i \) the \( i \)-th variable valuation along a run. Define \( Y_i\) by \( Y_i := \eta(S_i) + f(\veca{x}_i^\out)\) for any \( i \geq 0 \). We say that the tuple \( (\pCFG, \eta, f) \) is \emph{OST-sound} if \( \E[|Y_i|] < \infty \) for every \( i\geq 0 \) and moreover, at least one of the following conditions (C1)--(C4) holds:
	\begin{compactitem}
		\item[(C1)] There exists a constant \( c \) such that \( \TimeTerm \leq c \) with probability 1.
		\item[(C2)] There exists a constant \( c \) such that for each \( t \in \Nset \) and run \( \rho \) we have $|Y_{\min\{t, \TimeTerm(\rho)\}}(\rho)| \leq c$ (i.e., \( Y_{i}(\rho) \) is uniformly bounded along the run).
		\item[(C3)] \( \E[\TimeTerm] < \infty \) and there exists a constant \( c \) such that for every \( t \in \Nset \) it holds \( \E[|Y_{t+1}-Y_{t}|\mid \sigmaAlg_t] \leq c \) (i.e., the expected one-step change of \( Y_i \) is uniformly bounded over runtime).
		\item [(C4)] There exist real numbers \( M, c_1, c_2, d \) such that (i) for all sufficiently large \( n \in \Nset \) it holds \( \probm(\TimeTerm > n) \leq c_1 \cdot e^{-c_2 \cdot n} \); and (ii) for all \( t \in \Nset \) it holds \( |Y_{n+1} - Y_n| \leq M\cdot n^d \).
	\end{compactitem}
\end{definition}

We now present our new proof rule for equivalence refutation in PPs with conditioning. Given a pair \( P^1,P^2 \) of a.s.~terminating PPs that satisfy the bounded cumulative weight property, our proof rule applies the weighted restarting procedure to produce the pCFGs \( \pCFG(\Prestarted^1) \) and \( \pCFG(\Prestarted^2) \) and then applies the proof rule of~\cite{chatterjee2024equivalence} for PPs without conditioning. The following theorem formalizes the instantiation of the proof rule of~\cite{chatterjee2024equivalence} to the pCFG pair \( \pCFG(\Prestarted^1) \) and \( \pCFG(\Prestarted^2) \).

\begin{theorem}[Equivalence refutation for PPs with conditioning]
\label{thm:newproofrule}
Let \( P^1,P^2 \) be two a.s. terminating programs (possibly with conditioning) satisfying the bounded cumulative weight property. Let the initial states of \( \pCFG(\Prestarted^1) \) and \( \pCFG(\Prestarted^2) \) be \( (\locinit^1,\vecinit^1) \) and \( (\locinit^2,\vecinit^2) \), respectively.
Assume that there exist a Borel-measurable function \( f \colon \Rset^{|V_\out|} \rightarrow \Rset \) and two state functions, \( U_f \) for \( \pCFG(\Prestarted^1) \) and \( L_f \) for \( \pCFG(\Prestarted^2) \) such that the following holds:
\begin{compactitem}
\item[i)] \( U_f \) is a UESM for \( f \) in \( \pCFG(\Prestarted^1) \) such that \( (\pCFG(\Prestarted^1), U_f, f) \) is OST-sound;
\item[ii)] \( L_f \) is an LESM for \( f \) in \( \pCFG(\Prestarted^2) \) such that \( (\pCFG(\Prestarted^2), L_f, f) \) is OST-sound;
\item[iii)] \( U_f(\locinit^1, \vecinit^1) + f((\vecinit^1)^\out) < L_f(\locinit^2, \vecinit^2) + f((\vecinit^2)^\out) \).
\end{compactitem}
Then \( P^1,P^2 \) are not output equivalent,
%
%Moreover, if \( f \) is \( 1 \)-Lipschitz continuous under a metric \( d \) over \( \Rset^{|V_\out|} \), then 
%$$\K_d\Big(\mu^{\pCFG(P^1)},\mu^{\pCFG(P^2)}\Big) \geq L_f(\locinit^2, \vecinit^2) + f((\vecinit^2)^\out) - U_f(\locinit^1, \vecinit^1) - f((\vecinit^1)^\out).$$
\end{theorem}

%\begin{proof}
%If \( P^1, P^2 \) are a.s. terminating, then so are \(  \Prestarted^1, \Prestarted^2 \). The rest follows directly from Theorems~\ref{thm:ulesm-refute} and~\ref{thm:newWhat proofrule}.
%\end{proof}

The theorem proof follows immediately by the correctness of weighted restarting (Theorem~\ref{thm:weightedrestarting}) and the correctness of the proof rule in~\cite[Theorem~5.5]{chatterjee2024equivalence}. \footnote{We remark that the program syntax considered in~\cite{chatterjee2024equivalence} does not include go-to statements. However, the definition of U/LESMs and the proof rule in~\cite{chatterjee2024equivalence} are formalized with respect to pCFGs. Each go-to statement in a program can simply be modeled as a new edge in the program’s pCFG. Since~\cite{chatterjee2024equivalence} does not make any restriction on the topology of pCFGs, their result remains correct even for pCFGs induced by conditioning-free PPs with go-to statements.}

While~\cite{chatterjee2024equivalence} establishes soundness of the above proof rule for PPs {\em without conditioning}, which here we generalize to PPs {\em with conditioning}, to the best of our knowledge there is no known proof rule that is also complete. In what follows, we show that the proof rule in Theorem~\ref{thm:newproofrule} is not only sound, but also complete. This result is new both in the setting of PPs with and without conditioning.

%the equivalence refutation proof rules in both Theorems~\ref{thm:proofrule} and~\ref{thm:newproofrule} are not only sound, but also complete. We proof this for the new proof rule, since it strictly generalizes the old one.

\begin{theorem}[Completeness of the proof rule]\label{thm:completeness}
Let \( P^1, P^2 \) be a two a.s. terminating programs satisfying the bounded cumulative weight property. Assume that they are not output equivalent. Then there exists a Borel-measurable function \( f \colon \Rset^{|V_\out|} \rightarrow \Rset \) and two state functions, \( U_f \) for $f$ in \( \pCFG(\Prestarted^1) \) and \( L_f \) for $f$ in \( \pCFG(\Prestarted^2) \), satisfying the conditions of Theorem~\ref{thm:newproofrule}.
\end{theorem}

\begin{proof}
Since \( \mu^{\pCFG(P^1)} \neq \mu^{\pCFG(P^2)}\), by Theorem~\ref{thm:weightedrestarting} also \( \mu^{\pCFG(\Prestarted^1)} \neq \mu^{\pCFG(\Prestarted^2)}\). Hence, there is a Borel-measurable set \( B \subseteq \mathbb{R}^{\Vout} \) such that, w.l.o.g., \( \mu^{\pCFG(\Prestarted^1)}[B] > \mu^{\pCFG(\Prestarted^2)}[B] \). Let $p^i_B(\loc,\vec{x}) = \probm^{\pCFG(\Prestarted^i)}[\Output(B) \mid \text{run initialized in } (\loc,\vec{x})]$ be the probability of outputting a vector from \( B \) if \( \pCFG(\Prestarted^i) \) starts in \( (\loc,\vec{x}) \), for $i \in \{1,2\}$. We define:
\begin{compactitem}
\item A function \( f = \indicator{B} \) over outputs;
\item A state function \( U_f(\loc,\vec{x}) = p^1_B(\loc,\vec{x}) - f(\vec{x}^\out) \) in \( \pCFG(\Prestarted^1) \); %where
%\[ p^1_B(\loc,\vec{x}) = \probm^{\pCFG(\Prestarted^1)}[\Output(B) \mid \text{run initialized in } (\loc,\vec{x})] \]
%is the probability of outputting a vector from \( B \) if \( \pCFG(\Prestarted^1) \) starts in \( (\loc,\vec{x}) \).
\item A state function \( L_f(\loc,\vec{x}) = p^2_B(\loc,\vec{x}) - f(\vec{x}^\out) \) in \( \pCFG(\Prestarted^2) \).%, where \( p^2_B(\loc,\vec{x}) \) is defined as above.%, but in \( \pCFG(\Prestarted^2) \) instead of \( \pCFG(\Prestarted^1) \).
\end{compactitem}
We prove that \( U_f \) is an UESM $f$ in \( \pCFG(\Prestarted^1) \), the argument for \( L_f \) is analogous.
For any terminal state \( (\loc,\vec{x}) \) it holds \( p^1_b(\loc,\vec{x}) = \indicator{B}(\vec{x}^\out) = f(\vec{x}^\out) \), hence \( U_f \) is zero on output. Moreover, for any state \( (\loc,\vec{x}) \) it holds \( U_f(\loc,\vec{x}) + f(\vec{x}^\out) = p^1_B(\loc,\vec{x}) \). Since \( p^1_B(\loc,\vec{x}) = \mathbb{E}[p^1_B(\loc',\mathbf{N})]  \) (the standard flow conservation of reachability probabilities), we get that \( U_f \) satisfies the expected \( f \)-decrease condition. 

Regarding OST soundness, note that \( U_f(\loc,\vec{x}) + f(\vec{x}^\out) = p_B^1\in [0,1] \) for any state \( (\loc,\vec{x}) \). Hence, \( (\pCFG(\Prestarted^1),U_f,f) \) is OST-sound -- it satisfies condition (C2) -- and similarly for \( (\pCFG(\Prestarted^2),L_f,f) \).

It remains to verify condition iii) from Theorem~\ref{thm:newproofrule}. We have
\begin{align*}
U_f(\locinit^1, \vecinit^1) + f((\vecinit^1)^\out) &= p^1_B(\locinit,\vecinit) = \mu^{\pCFG(\Prestarted^1)}[B] < \mu^{\pCFG(\Prestarted^2)}[B] \\ &= p^2_B(\locinit,\vecinit) = L_f(\locinit^1, \vecinit^1) + f((\vecinit^1)^\out).
\end{align*}
\end{proof}

\begin{example}
	To conclude this section, we illustrate our proof rule on our running example in Fig.~\ref{fig:running}. Consider the PPs $(\Prestarted^1,\Prestarted^2)$ obtained by weighted restarting in Fig.~\ref{fig:restarting}, and the function $f$ over outputs, the UESM $U_f$ for $f$ in $\pCFG(\Prestarted^1)$, and the LESM $L_f$ for $f$ in $\pCFG(\Prestarted^2)$ constructed in Example~\ref{ex:uesm}. The triples $(\pCFG(\Prestarted^1),U_f,f)$ and $(\pCFG(\Prestarted^2),L_f,f)$ are both OST-sound and satisfy (C2) in Definition~\ref{def:ostsound}, as all variable values in $\Prestarted^1$ and $\Prestarted^2$ are bounded. Moreover, we have $U_f(\locinit^1, \vecinit^1) + f((\vecinit^1)^\out) = 500 < 600 = L_f(\locinit^2, \vecinit^2) + f((\vecinit^2)^\out)$. Hence, $(f,U_f,L_f)$ satisfy all conditions of Theorem~\ref{thm:newproofrule}, and PPs in Fig.~\ref{fig:running} are not output equivalent. 
\end{example}

\section{Automated Equivalence Refutation}\label{sec:algorithm}

We now present our algorithm for equivalence refutation in PPs that may contain conditioning instructions. Given a PP pair \( (P^1, P^2) \), the algorithm first applies weighted restarting to translate them into output equivalent conditioning-free PPs \( (\Prestarted^1, \Prestarted^2) \). Then, the algorithm applies the equivalence refutation procedure of~\cite{chatterjee2024equivalence} for PPs without conditioning. Hence, given that our algorithm builds on that of~\cite{chatterjee2024equivalence}, we only present an overview of the algorithm and refer the reader to~\cite{chatterjee2024equivalence} for details. However, this section presents two important novel results:
\begin{compactenum}
	\item {\em Semi-completeness guarantees.} Building on our result in Section~\ref{sec:proofrules} which shows that our proof rule is sound and complete, we show that our algorithm for equivalence refutation is sound and {\em semi-complete}. In particular, the algorithm is guaranteed to refute equivalence of PPs whenever non-equivalence can be witnessed by a function $f$ and a pair of U/LESMs that can be specified via polynomial functions of some maximal polynomial degree $D$, and that satisfy OST-soundness condition (C2) in Definition~\ref{def:ostsound}.
	\item {\em Streamlined algorithm for bounded termination time PPs.} We show that the algorithm can be streamlined for the class of PPs whose (1)~termination time is bounded by some constant value $T > 0$, and (2)~all sampling instructions consider probability distributions of bounded support. The first assumption is common and satisfied by most PPs considered in statistical inference literature, including all our benchmarks in Section~\ref{sec:experiments}.
\end{compactenum}

%In what follows, let $\pCFG_1=(\locs^1,V^1,\Vout,\locinit^1,\vecinit^1,\transitions^1,\guards^1,\updates^1,\locfail^1,\locterm^1)$ and $\pCFG_2=(\locs^2,V^2,\Vout,\locinit^2,\vecinit^2,\transitions^2,\guards^2,\updates^2,\locfail^2,\locterm^2)$ be two pCFGs with a common set of output variables $\Vout$. Our algorithm builds on and extends the algorithm of~\cite{chatterjee2024equivalence} for refuting equivalence of PPs without conditioning instructions. Hence, we keep the exposition brief and defer the details to Appendix~\ref{app:algo}.

\smallskip\noindent{\em Algorithm assumptions.} In order to allow for a fully automated equivalence refutation, we need to impose several assumptions on our PPs:
\begin{compactitem}
	\item {\em Polynomial arithmetic.} We require all arithmetic expressions appearing in PPs to be polynomials over program variables. This restriction is necessary for the full automation of our method. We also require that each probability distribution appearing in sampling instructions in pCFGs  has {\em  finite moments}, i.e. for each \( p \in \Nset \), the \( p \)-th moment $m_{\delta}(p) = \mathbb{E}_{X \sim \delta}[|X|^p]$ is finite and can be computed by the algorithm. This is a standard assumption in static PP analysis and allows for most standard probability distributions.
	\item {\em Almost-sure termination.} For the equivalence refutation problem to be well-defined and for our proof rule to be applicable, we require both PPs to be a.s.~terminating. In polynomial arithmetic PPs, this requirement can be {\em automatically checked} by using the method of~\cite{ChatterjeeFG16}.
	\item {\em Bounded cumulative weight property.} To automatically check this property, we can use an off-the-shelf invariant generator, e.g.~\cite{FeautrierG10,SankaranarayananSM04,Chatterjee0GG20}, to generate an invariant which proves that $W \leq M$ holds for the weight variable $W$ and some sufficiently large constant value $M$.
	%\item {\em Finite-moment distributions.} We require that each probability distribution appearing in sampling instructions in pCFGs  has {\em  finite moments}, i.e. for each \( p \in \Nset \), the \( p \)-th moment $m_{\delta}(p) = \mathbb{E}_{X \sim \delta}[|X|^p]$ is finite and can be computed by the algorithm. This is a standard assumption in static PP analysis and allows for most standard probability distributions.
	\item {\em Supporting linear invariants.} Recall that we defined UESMs and LESMs with respect to supporting invariants. To that end, before proceeding to equivalence refutation, our algorithm first synthesizes {\em linear invariants} $I_1$ and $I_2$. This can be fully automated by using existing linear invariant generators~\cite{FeautrierG10,SankaranarayananSM04}, as done in our implementation. For the purpose of invariant generation, sampling instructions are semantically over-approximated via non-determinism, hence ensuring that the computed invariants are sound.
\end{compactitem}

\smallskip\noindent{\em Template-based synthesis algorithm.} Our algorithm first uses weighted restarting to produce \( \Prestarted^1 \) and \( \Prestarted^2 \), and constructs their pCFGs \( \pCFG(\Prestarted^1) \) and \( \pCFG(\Prestarted^2) \). It then follows a template-based synthesis approach and simultaneously synthesizes the triple $(f, U_f, L_f)$ required by our proof rule in Theorem~\ref{thm:newproofrule}. The algorithm proceeds in four steps. In Step~1, it fixes templates for $f$, $U_f$ and $L_f$, where the templates are symbolic polynomials over pCFG variables of some maximal polynomial degree $D$, an algorithm parameter provided by the user. In Step~2, the algorithm collects a system of constraints over the symbolic polynomial template variables, which together encode that $(f, U_f, L_f)$ satisfy all the requirements and give rise to a correct instance of the proof rule in  Theorem~\ref{thm:newproofrule}. At this point, the resulting constraints contain polynomial expressions and quantifier alternation, which makes their solving challenging. Thus, in Step~3, the algorithm translates the collected system of constraints into a purely {\em  existentially quantified} system of {\em linear} constraints over the symbolic template variables. This is done via Handelman's theorem~\cite{handelman1988representing}, analogously to previous template-based synthesis PP analyses~\cite{Wang0GCQS19,AsadiC0GM21,ZikelicCBR22}. In Step~4, the algorithm solves the system of linear constraints via a linear programming (LP) tool. If the system is feasible, its solution yields a concrete instance of $(f, U_f, L_f)$ from Theorem~\ref{thm:newproofrule} and  the algorithm reports ``Not output-equivalent''. Due to reducing the synthesis to an LP instance, the runtime complexity of our algorithm is polynomial in the size of the pCFGs, for any fixed value of the polynomial degree parameter $D$. Finally, the proof of Theorem~\ref{thm:completeness} establishes completeness of our proof rule with (C2) OST-soundness condition, and here we show semi-completeness of our algorithm. 

\begin{theorem}[Algorithm correctness, complexity and semi-completeness, proof in Appendix]\label{thm:algo}
	The algorithm is correct. If the algorithm outputs ''Not output-equivalent'', then \( (P^1, P^2) \) are indeed not output-equivalent and the computed triple $(f, U_f, L_f)$ forms a correct instance of the proof rule in Theorem~\ref{thm:newproofrule}.
	
	Furthermore, the runtime complexity of the algorithm is polynomial in the size of the programs' pCFGs, for any fixed value of the polynomial degree $D$.
	
	Finally, the algorithm is semi-complete, meaning that if there exists a valid triple $(f, U_f, L_f)$ that can be specified via polynomial expressions of degree at most $D$ and that satisfy OST-soundness condition (C2) in Definition~\ref{def:ostsound}, then the algorithm is guaranteed to output ''Not output-equivalent''.
\end{theorem}

\noindent{\em Streamlined algorithm for bounded termination time PPs.} We conclude this section by showing that the algorithm can be streamlined for the class of PPs whose (1)~termination time is bounded by some $T > 0$, and (2)~all sampling instructions consider probability distributions of bounded support. We show that for this class of PPs, any polynomial function $f$ over outputs and any polynomial UESM $U_f$ (resp.~LESM $L_f$) for $f$ are guaranteed to be (C2) OST-sound. Hence, the algorithm does not need to check OST-soundness, which significantly simplifies the system of constraints that needs to be solved by the SMT-solver.

\begin{theorem}[Proof in Appendix] \label{thm:c2algo}
	Let $P$ be a PP with (1)~termination time bounded by some $T > 0$, and (2)~all sampling instructions considering probability distributions of bounded support. Then, for any polynomial function $f$ over outputs, polynomial UESM $U_f$ for $f$ in $\pCFG(\Prestarted^1)$ and LESM $L_f$ for $f$ in $\pCFG(\Prestarted^2)$, $(\pCFG(\Prestarted^1),U_f,f)$ and $(\pCFG(\Prestarted^2),L_f,f)$ satisfy the (C2) OST-soundness condition.
\end{theorem}
\section{Experimental Results}\label{sec:experiments}

% For a given PP pair, our prototype tool tries synthesizing an instance of the proof rule from Theorem~\ref{thm:proofrule} via the algorithm presented in Section~\ref{sec:algorithm}. %Given two probabilistic programs $P_1$ and $P_2$ with the same output variables $\Vout$ that might contain conditioning (in the form of \texttt{observe} statements), the tool tries to prove that their output distributions are not the same. 

\noindent{\em Implementation.} We implemented a prototype of our equivalence refutation algorithm. The prototype tool is implemented in Java. Gurobi \cite{gurobi} is used for solving the LP instances while STING \cite{SankaranarayananSM04} and ASPIC \cite{FeautrierG10} are used for generating supporting invariants. For each input program pair, the tool tries to perform polynomial template-based synthesis by using polynomial degrees 1 to 5 in order. All experiments were conducted on an Ubuntu 24.04 machine with an 11th Gen Intel Core i5 CPU and 16 GB RAM with a timeout of 10 minutes.

\smallskip\noindent{\em Baseline.} To the best of our knowledge, no prior work has considered formal and automated equivalence refutation for PPs with conditioning. However, a viable alternative approach to refuting PP equivalence would be to first compute the output distributions of the two PPs via exact symbolic inference, and then to use mathematical computation software to show that these distributions are different. Hence, we compare our method against such a baseline. We use the state-of-the-art exact inference tool PSI \cite{DBLP:conf/cav/GehrMV16} for computing output distributions, and Mathematica \cite{Mathematica} for proving that they are different.

\begin{table}
	\caption{Experimental results showing the performance of our tool and the baseline for refuting equivalence of PPs. A $\checkmark$ in the ``Eq. Ref.'' column represents that the tool successfully refuted equivalence of the two input programs, whereas ``TO'' and ``NA'' stand for ``Timeout'' and ``Not Applicable'', respectively. The \texttt{Distance} column shows the Kantorovich distance lower bound computed by out tool for the similarity refutation problem (see Section~\ref{sec:intro}).}
		%(1) comparison of our equivalence refutation method and the baseline, (2) lower bounds on Kantorovich distance computed by our similarity refutation method, and (3) time taken to solve each instance. A $\checkmark$ in the ``Eq. Ref.'' column represents that the tool successfully refuted equivalence of the two input programs, ``TO'' and ``NA'' stand for ``timeout'' and ``Not Applicable'', respectively.}
%\vspace{-1em}
\centering
\texttt{
	\resizebox{\textwidth}{!}
	{
		\fontsize{4pt}{5pt}\selectfont
		\begin{tabular}[h]{|c|c|c c c c|c c|}
			\hline
			& \multirow{2}{*}{Name} & \multicolumn{4}{c|}{Our Method} & \multicolumn{2}{c|}{PSI + Mathematica}\\
			& & Eq. Ref. & Time(s) & Distance & Time(s) & Eq. Ref. & Time(s) \\
			\hline 
			\multirow{12}{*}{\rotatebox[origin=c]{90}{Benchmarks from \cite{DBLP:conf/cav/GehrMV16}}} & 
			bertrand & \checkmark & 0.22 & 0.42 & 0.30 & \checkmark & 0.86 \\
			\cline{2-8}	
			& burglarAlarm & \checkmark & 2.01 & 0.05 & 1.39 & \checkmark & 1.04 \\
			\cline{2-8}	
			& coinBiasSmall & \checkmark & 36.92 & TO & - & \checkmark & 1.07\\
			\cline{2-8}	
			& coinPattern & \checkmark & 4.50 & 0.04 & 56.33 & \checkmark & 0.86\\
			\cline{2-8}	
			& coins & \checkmark & 0.46 & 0.12 & 0.95 & \checkmark & 0.83\\
			\cline{2-8}	
			& ev-model1 & \checkmark & 0.45 & 0.16 & 0.81 &\checkmark & 0.85\\
			\cline{2-8}	
			& ev-model2 & \checkmark & 0.75 & 0.20 & 0.92 & \checkmark & 0.80\\
			\cline{2-8}	
			& gossip & \checkmark & 0.62 & 0.1 & 0.63 & \checkmark & 1.01\\
			\cline{2-8}
			& grass & \checkmark & 1.46 & 0.06 & 8.15 & \checkmark & 1.02\\
			\cline{2-8}
			& murderMystery & \checkmark & 0.53 & 0.004 & 1.24 &\checkmark & 0.84\\
			\cline{2-8}
			& noisyOr & \checkmark & 1.17 & 0.17 & 2.15 & \checkmark & 1.98\\
			\cline{2-8}
			& twoCoins & \checkmark & 0.66 & 0.12 & 1.86 & \checkmark & 0.93\\
			\hline
			\hline 
			\multirow{18}{*}{\rotatebox[origin=c]{90}{Benchmarks from \cite{chatterjee2024equivalence}}} & 
			Simple Example & \checkmark & 7.28 & 381.71 & 12.15 & TO & - \\
			\cline{2-8}	
			& Nested Loop & TO & - & TO & - & TO & - \\
			\cline{2-8}	
			& Random Walk & \checkmark & 32.01 & TO & - & TO & - \\
			\cline{2-8}	
			& Goods Discount & \checkmark & 0.76 & 0.08 & 1.05 & TO & - \\
			\cline{2-8}	
			& Pollutant Disposal & \checkmark & 3.08 & 327.27 & 3.43 & TO & -\\
			\cline{2-8}	
			& 2D Robot & \checkmark & 0.77 & TO & - & TO & -\\
			\cline{2-8}	
			& Bitcoin Mining & \checkmark & 7.28 & TO & - & TO & - \\
			\cline{2-8}	
			& Bitcoin Mining Pool & \checkmark & 13.248 & TO & - & TO & - \\
			\cline{2-8}	
			& {Species Fight} & TO & - & TO & - & TO & -\\
			\cline{2-8}	
			& {Queuing Network} & TO & - & TO & - & TO & -\\
			\cline{2-8}	
			& {Coupon1} & \checkmark & 3.86 & 0.47 & 4.64 & TO & -\\
			\cline{2-8}	
			& {Coupon4} & \checkmark & 169.47 & TO & - & TO & -\\
			\cline{2-8}	
			& {random\_walk\_1d\_intvalued} & \checkmark & 6.24 & TO & - & TO & -\\
			\cline{2-8}	
			& {random\_walk\_1d\_realvalued} & TO & - & TO & - & TO & -\\
			\cline{2-8}	
			& {random\_walk\_1d\_adversary} & TO & - & TO & - & TO & -\\
			\cline{2-8}	
			& {random\_walk\_2d\_demonic} & \checkmark & 25.99 & TO & - & TO & -\\
			\cline{2-8}	
			& {random\_walk\_2d\_variant} & TO & - & TO & - & TO & -\\
			\cline{2-8}	
			& {retransmission} & \checkmark & 2.43 & TO & - & TO & -\\
			%\cline{1-8}
			\hline		
			\hline
			\multirow{11}{*}{\rotatebox[origin=c]{90}{Benchmarks from \cite{WangYFLO24}}}
			& add-uni & \checkmark & 10.25 & 0.01 & 17.48 & TO & - \\
			\cline{2-8}
			& cav-ex-5 & TO & - & TO & - & TO & - \\
			\cline{2-8}
			& cav-ex-7 & \checkmark & 0.19 & 10.00 & 0.20 & TO & - \\
			\cline{2-8}
			& pdmb-v3 & \checkmark & 60.27 & 2100.00 & 60.27 & NA & - \\
			\cline{2-8}
			& race & \checkmark & 1.40 & 2350.42 & 1.39 & NA & - \\
			\cline{2-8}
			& rdwalk-v12 & \checkmark & 0.85 & 100.00 & 0.75 & NA & - \\
			\cline{2-8}
			& rdwalk-v13 & \checkmark & 1.36 & 0.01 & 16.53 & NA & - \\
			\cline{2-8}
			& rdwalk-v14 & \checkmark & 0.77 & 100.00 & 0.80 & NA & - \\
			\cline{2-8}
			& rdwalk-v23 & \checkmark & 0.78 & 100.00 & 0.90 & NA & - \\
			\cline{2-8}
			& rdwalk-v24 & \checkmark & 1.06 & 200.00 & 0.90 & NA & - \\
			\cline{2-8}
			& rdwalk-v34 & \checkmark & 2.19 & 100.00 & 1.09 & NA & - \\
			\hline \hline 
			%\cline{1-8}
			\multicolumn{2}{|c|}{Success Count } & 34 & - & 26 & - & 12 & - \\
			\multicolumn{2}{|c|}{Average Runtime} & - & 11.8 & - & 7.63 & - & 1.01 \\
			\hline 
			%\cline{1-8}
		\end{tabular}
	}
}
\label{tab:experiments}
\end{table}

\smallskip\noindent{\em Benchmarks.} While our method is applicable to PPs with unbounded loops, PSI (and hence our baseline) is restricted to PPs with statically bounded loops. Hence, for all our benchmarks, we bound the maximal number of iterations of each loop by $1000$. Note that, while this modification ensures that all loops admit an upper bound on the number of iterations, it still allows loops to terminate early (i.e.~in strictly less than $1000$ iterations) if the termination condition is met early. Hence, executions of programs are not aligned by this modification, making them challenging in structure for relational analysis.

We ran our tool and the baseline on three sets of benchmarks. For each benchmark in the first two sets, we obtain a pair of non-output equivalent PPs by slightly perturbing exactly one of the sampling instructions appearing in the original PP. Benchmarks in the third set already come in pairs:
\begin{compactenum}
	\item We consider $12$ PPs with discrete probability distributions from the PSI benchmark suite (taken from the repository of~\cite{BeutnerOZ22}). We selected benchmarks supported by our syntax, e.g., avoiding reading data from external \textsc{.csv} files.
	\item We consider $18$ loopy programs from~\cite{chatterjee2024equivalence} on equivalence refutation in PPs without conditioning (benchmarks originating from \cite{Wang0GCQS19,KuraUH19}). In order to introduce conditioning statements, to each program we add an \texttt{observe} statement at the end, which observes a bound on one of the output variables. 
	\item Lastly, we consider the benchmarks used in~\cite{WangYFLO24} for symbolic inference in PPs with conditioning. Specifically, we consider pairs of benchmarks that already have several versions in the repository of \cite{WangYFLO24} (for example there are four versions of the \texttt{rdwalk} program that were different only in one parameter; we ran our program on all 6 possible program pairs). All the programs in this benchmark set contain \texttt{score} statements. Note that PSI (and hence the baseline) does not support \texttt{score} statements. However, some benchmarks from this repository contain \texttt{"if (g) score(1) else score(0) fi"} fragments, which we replace by $\texttt{observe(g)}$ in order to make PSI applicable.
\end{compactenum}

\medskip\noindent{\em Results.} Our experimental results are shown in Table~\ref{tab:experiments}. We observe the following:
\begin{compactitem}
\item On the first benchmark set in which PPs terminate in a small number of steps, both our method and the baseline manage to refute equivalence of all $12$ PP pairs. Our method is faster in $8/12$ cases, whereas the baseline is faster in $4/12$ cases. However, in some of these cases our method is significantly slower (e.g.~\texttt{coinBiasSmall}), because the synthesis algorithm needed to use degree~$5$ polynomial templates to refute equivalence. However, our method overall shows a better performance than the baseline. Moreover, our method successfully computes Kantorovich distance lower bounds in $11/12$ cases.
\item On the second benchmark set in which PPs take much longer time to terminate (e.g~loops may be executed for up to $1000$ iterations), our method refutes equivalence in $12/18$ cases, whereas the baseline fails in all cases. This shows an advantage of our method. The baseline is based on exact symbolic inference, meaning that it needs to symbolically execute the program which is a computational bottleneck when the program runtime is large. In contrast, our method is a static analysis method and does not suffer from this limitation. When it comes to computing lower bounds on Kantorovich distance, our tool succeeds in only $4/18$ cases. We believe this is due to (i) the invariant generation tool StIng failing to compute invariants in several cases, or (ii) the UESM/LESMs needed being non-polynomial (e.g. for \texttt{Species Fight}).
\item On the third benchmark set, our tool successfully refutes equivalence and computes lower bounds on distance in $10/11$ cases. These results demonstrate the applicability of our method in the presence of both \texttt{observe} and \texttt{score} statements. The baseline solves no cases -- fails on $3$ whereas it is not applicable in $8$ cases as PSI does not support \texttt{score} statements.

\end{compactitem}

\smallskip\noindent{\em Limitations.} In our experiments, we also observed a limitation of our algorithm, addressing which is an interesting direction of future work. Supporting invariants play an important role for guiding the U/LESM synthesis algorithm. Without good supporting invariants, the synthesis algorithm might fail. Generating good invariants can sometimes be time consuming. For example, our tool takes 169 seconds to solve the \texttt{coupon4} benchmark in our experiments, because STING and ASPIC invariant generators take more than 2 minutes to generate an invariant.

%\smallskip\noindent{\em Limitations.} Although our method introduces theoretical guarantees and practical advancements, we note two main limitations:
%\begin{compactenum}
%\item Similar to many previous template-based methods \cite{chatterjee2024equivalence,Chatterjee0GG20}, our algorithm can only find polynomial witnesses for refuting equivalence. Thus, if non-polynomial, e.g. exponential, functions are required for refuting equivalence of a program pair, our method is not applicable. 
%\item Supporting invariants play an important role for guiding the U/LESM synthesis algorithm and without strong invariants, the algorithm might fail. Generating such invariants can sometimes be time-consuming. For example, our tool takes 169 seconds for solving the \texttt{coupon4} benchmark in our experiments because STING and ASPIC take more than 2 minutes for generating invariants.
%\end{compactenum}
\section{Conclusion}\label{sec:conclusion}

We presented a new method for refuting equivalence of PPs with conditioning. To the best of our knowledge, we presented the first method that is fully automated, provides formal guarantees, and is sound in the presence of conditioning instructions. Moreover, we show our proof rule is also complete and our algorithm is semi-complete. Our prototype implementation demonstrates the applicability of our method to a range of pairs of PPs with conditioning. An interesting direction of future work would be to consider the automation of our method for non-polynomial programs. Another interesting direction would be to consider the applicability of the weighted restarting procedure to other static analyses in PPs with conditioning, towards reducing them to the conditioning-free setting.

\section*{Acknowledgments}
This work was partially supported by ERC CoG 863818 (ForM-SMArt) and Austrian Science Fund (FWF) 10.55776/COE12. Petr Novotný is supported by the Czech Science Foundation grant no. GA23-06963S.

\bibliographystyle{splncs04}
\bibliography{bibliography}

\newpage
\appendix
\begin{center}
	\Large
	Technical Appendix
\end{center}

\section{Proof of Theorem~\ref{thm:weightedrestarting}}

\begin{theorem}[Correctness of weighted restarting]
	Consider a PP $P$ whose pCFG is assumed to be a.s.~terminating, integrable, and to satisfy the bounded cummulative weight property with bound $M>0$. Let $\Prestarted$ be a conditioning-free PP produced by the weighted restarting procedure. Then, $\Prestarted$ is also a.s.~ terminating, integrable, and satisfies the bounded cummulative weight property. Moreover, the pCFGs of $P$ and $\Prestarted$ are output equivalent.
\end{theorem}

\begin{proof}
	Let $\pCFG$ and $\pCFGrestarted$ be the pCFGs of $P$ and $\Prestarted$, respectively. Let $(\Run^\pCFG, \mathcal{F}^\pCFG, \probm^\pCFG)$ and $(\Run^{\pCFGrestarted}, \mathcal{F}^{\pCFGrestarted}, \probm^{\pCFGrestarted})$ be the probability spaces over the sets of pCFG runs defined by the semantics of the two pCFGs.
	
	\smallskip\noindent{\em Well-definedness of output equivalence.} In order for the notion of output equivalence of the two pCFGs to be well defined, we first need to show that $\pCFGrestarted$ is also a.s.~terminating, integrable and satisfies the bounded cummulative weight property. Since $\pCFGrestarted$ is conditioning-free, the cumulative weight of each terminating run in it is $1$. Hence, $\pCFGrestarted$ is integrable and satisfies the bounded cummulative weight property.
	
	On the other hand, to prove that $\pCFGrestarted$ is a.s.~terminating, observe that a run in  $\pCFGrestarted$ does not terminate only if it can be decomposed into one of the following two concatenation forms:
	\begin{compactenum}
		\item $\rho = \rho_1 \circ \rho_2 \circ \dots$, which is a concatenation of infinitely many finite cycles from the initial state to itself, whose last transition corresponds to the "\textbf{go-to} $\locinit$" restarting command introduced by the weighted restarting procedure;
		\item $\rho = \rho_1 \circ \dots \circ \rho_k \circ \rho^{\textit{nonterm}}$, where $\rho_1,\dots,\rho_k$ are again finite cycles from the initial state to itself whose last transition corresponds to the "\textbf{go-to} $\locinit$" restarting command, and $\rho^{\textit{nonterm}}$ is a non-terminating run in $\pCFGrestarted$ that never executes the "\textbf{go-to} $\locinit$" restarting transition.
	\end{compactenum}
	Denote the sets of all non-terminating runs of each type as $A_1$ and $A_2$, respectively. We then have
	\[ \probm^{\pCFGrestarted}[\Termset] = 1 - \probm^{\pCFGrestarted}[A_1] - \probm^{\pCFGrestarted}[A_2]. \]
	Hence, to prove that $\pCFGrestarted$ is a.s.~terminating, it suffices to show that both $\probm^{\pCFGrestarted}[A_1] = 0$ and $\probm^{\pCFGrestarted}[A_2] = 0$.
	
	Let $\prestart$ denote the probability of a run in $\pCFGrestarted$ getting restarted. By the construction of $\pCFGrestarted$, we have that
	\begin{equation}\label{eq:dummy4}
		\prestart = 1 - \mathbb{E}^\pCFG[W] / M < 1,
	\end{equation}
	where the last inequality holds since $\mathbb{E}^\pCFG[W] > 0$ by the integrability of $\pCFG$. Now, since each execution of the restarting transition moves a run to the initial state, we have that all finite cycles of the non-terminating run $\rho = \rho_1 \circ \rho_2 \circ \dots$ are pairwise independent. Hence,
	\[ \probm^{\pCFGrestarted}[A_1] = \lim_{k\rightarrow\infty} \prestart^k = 0, \]
	where the last equality holds since $\prestart < 1$. On the other hand, all finite cycles as well as the last non-terminating run in $\rho = \rho_1 \circ \dots \circ \rho_k \circ \rho^{\textit{nonterm}}$ are also pairwise independent. Moreover, by the construction of $\pCFGrestarted$, we have that the probability of a run not terminating and not getting restarted is equal to the probability of non-termination in $\pCFG$. Hence,
	\begin{equation*}
	\begin{split}
		\probm^{\pCFGrestarted}[A_2] &= \sum_{k=0}^\infty \prestart^k \cdot \Big(1 - \probm^\pCFG[\Termset] \Big) \\
		&= \frac{1 - \probm^\pCFG[\Termset]}{1 - \prestart} = 0,
	\end{split}
	\end{equation*}
	where the last equality holds since $\probm^\pCFG[\Termset] = 1$ as $\pCFG$ is a.s.~terminating and since $\prestart < 1$. Hence $\probm^{\pCFGrestarted}[A_1] = 0$ and $\probm^{\pCFGrestarted}[A_2] = 0$ and so $\pCFGrestarted$ is a.s.~terminating, as claimed.
	
	\smallskip\noindent{\em Proving output equivalence of the pCFGs}. Now, to prove that the pCFGs of $P$ and $\Prestarted$ are output equivalent, we need to prove that the following holds for every Borel-measurable subset $B$ of $\mathbb{R}^{\Vout}$:
	\begin{equation}\label{eq:toprove}
		\mu^\pCFG(B) = \mu^{\pCFGrestarted}(B).
	\end{equation}
	
	%Before proving that eq.~\eqref{eq:toprove} holds for every $B$, we first introduce some additional notation. Let $\Paux$ denote an auxillary PP obtained from $P$ by applying all steps of the weighted restarting procedure except for "Adding restarting upon termination" (that is, only the first three steps in Section~\ref{sec:restarting} are applied). Denote by $\pCFGaux$ its pCFG, $(\Run^{\pCFGaux}, \mathcal{F}^{\pCFGaux}, \probm^{\pCFGaux})$ the probability space over its runs, and $\mu^{\pCFGaux}$ its normalized output distributions.
	
	To prove eq.~\eqref{eq:toprove}, let $B$ be a Borel-measurable subset of $\mathbb{R}^{\Vout}$. Since $\pCFGrestarted$ is conditioning-free by construction, we have that the cummulative weight of every run in $\pCFGrestarted$ is $1$, hence
	\begin{equation*}
		\mu^{\pCFGrestarted}(B) = \probm^{\pCFGrestarted}\Big[\Output^{\pCFGrestarted}(B)\Big],
	\end{equation*}
	where $\Output^{\pCFGrestarted}(B)$ denotes the set of all terminating runs in $\Run^{\pCFGrestarted}$ whose output variable valuation is in $B$. On the other hand, each run $\rho \in \Output^{\pCFGrestarted}(B)$ can be decomposed into a concatenation
	\[ \rho = \rho_1 \circ \dots \circ \rho_k \circ \rho^{\textit{term}} \]
	where:
	\begin{compactenum}
		\item $\rho_1, \dots, \rho_k$ are all finite cycles in $\Run^{\pCFGrestarted}$ from the initial state to itself, whose last transition corresponds to the "\textbf{go-to} $\locinit$" restarting command introduced by the weighted restarting procedure, and
		\item $\rho^{\textit{term}}$ is a terminating run in $\Run^{\pCFGrestarted}$ that does not contain the transition corresponding to the "\textbf{go-to} $\locinit$" restarting command, and whose output variable valuation is in $B$.
	\end{compactenum}
	This is because each terminating run in $\pCFGrestarted$ can get "restarted" only finitely many times before eventually terminating. Hence, we have
	\begin{equation}\label{eq:dummy2}
	\begin{split}
		&\mu^{\pCFGrestarted}(B) = \probm^{\pCFGrestarted}\Big[\Output^{\pCFGrestarted}(B)\Big] \\
		&= \probm^{\pCFGrestarted}\Big[\bigcup_{k=0}^\infty\Big\{ \rho \in \Output^{\pCFGrestarted}(B) \mid \rho = \rho_1 \circ \dots \circ \rho_k \circ \rho^{\textit{term}} \Big\} \Big] \\
		&= \sum_{k=0}^\infty \probm^{\pCFGrestarted}\Big[\Big\{ \rho \in \Output^{\pCFGrestarted}(B) \mid \rho = \rho_1 \circ \dots \circ \rho_k \circ \rho^{\textit{term}} \Big\} \Big] \\
		&= \sum_{k=0}^\infty \probm^{\pCFGrestarted}[B_k],
	\end{split}
	\end{equation} 
	where $B_k = \{ \rho \in \Output^{\pCFGrestarted}(B) \mid \rho = \rho_1 \circ \dots \circ \rho_k \circ \rho^{\textit{term}} \}$. Now, let $\Restart^{\pCFGrestarted}$ denote the set of all runs in $\Run^{\pCFGrestarted}$ that get restarted at least once, and $\neg\Restart^{\pCFGrestarted}$ the set of all runs in $\Run^{\pCFGrestarted}$ that never get restarted. Since each execution of the restarting transition moves a run to the initial state and since all program runs are independent, we can observe that
	\begin{equation*}
	\begin{split}
		\probm^{\pCFGrestarted}[B_k] &= \probm^{\pCFGrestarted}\Big[\Restart^{\pCFGrestarted}\Big]^k \\
		&\quad\cdot \probm^{\pCFGrestarted}\Big[\neg \Restart^{\pCFGrestarted} \cap \Output^{\pCFGrestarted}(B)\Big] \\
		&= \prestart^k \cdot \probm^{\pCFGrestarted}\Big[\neg \Restart^{\pCFGrestarted} \cap \Output^{\pCFGrestarted}(B)\Big]
	\end{split}
	\end{equation*}
	Plugging this into eq.~\eqref{eq:dummy2} yields
	\begin{equation}\label{eq:dummy3}
		\begin{split}
			\mu^{\pCFGrestarted}(B) &= \sum_{k=0}^\infty \probm^{\pCFGrestarted}[B_k] \\
			&= \sum_{k=0}^\infty \prestart^k \cdot \probm^{\pCFGrestarted}\Big[\neg \Restart^{\pCFGrestarted} \cap \Output^{\pCFGrestarted}(B)\Big] \\
			&= \frac{\probm^{\pCFGrestarted}\Big[\neg \Restart^{\pCFGrestarted} \cap \Output^{\pCFGrestarted}(B)\Big]}{1 - \prestart} \\
			&= \frac{\probm^{\pCFGrestarted}\Big[\neg \Restart^{\pCFGrestarted} \cap \Output^{\pCFGrestarted}(B)\Big]}{\mathbb{E}^\pCFG[W] / M},
		\end{split}
	\end{equation}
	where the last equality follows from the fact that $\prestart = 1 - \mathbb{E}^\pCFG[W] / M$ as shown in eq.~\eqref{eq:dummy4}. Finally, we have
	\begin{equation*}
	\begin{split}
		&\probm^{\pCFGrestarted}\Big[\neg \Restart^{\pCFGrestarted} \cap \Output^{\pCFGrestarted}(B)\Big] \\
		&= \mathbb{E}^{\pCFGrestarted}\Big[\mathbb{I}_{\neg \Restart^{\pCFGrestarted} \cap \Output^{\pCFGrestarted}(B)}\Big] \\
		&= \mathbb{E}^{\pCFG}\Big[\mathbb{I}_{\Output^{\pCFG}(B)} \cdot W/M\Big] = \mathbb{E}^{\pCFG}\Big[\mathbb{I}_{\Output^{\pCFG}(B)} \cdot W\Big] / M.
	\end{split}
	\end{equation*}
	Plugging these into eq.~\eqref{eq:dummy3} gives
	\[ \mu^{\pCFGrestarted}(B) = \frac{\mathbb{E}^{\pCFG}\Big[\mathbb{I}_{\Output^{\pCFG}(B)} \cdot W\Big] / M}{\mathbb{E}^{\pCFG}[W] / M} = \frac{\mathbb{E}^{\pCFG}\Big[\mathbb{I}_{\Output^{\pCFG}(B)} \cdot W \Big]}{\mathbb{E}^{\pCFG}[W]} = \mu^\pCFG(B) \]
	where the last equality holds by the definition of the normalized output distribution, which is exactly the claim of eq.~\eqref{eq:toprove}. This concludes our proof that the pCFGs $\pCFG$ and $\pCFGrestarted$ are output equivalent. \hfill\qed
\end{proof}

\section{Similarity Refutation}\label{app:similarity}

We now present our proof rule for similarity refutation in PPs with conditioning.

\subsection{Similarity Refutation Problem}

\smallskip\noindent{\em Kantorovich distance.} Following~\cite{chatterjee2024equivalence}, we use the established \emph{Kantorovich distance}~\cite{villani2021topics} to measure the similarity of NODs of two programs. Let \( d \colon \Rset^{|\Vout|} \times \Rset^{|\Vout|} \rightarrow [0,\infty) \) be any metric. We say that a probability measure \( \mu \) over \( \Vout \) has a \emph{finite first moment} w.r.t. \( d \) if there exists a point \( \mathbf{x}_0\in \Rset^{|\Vout|} \) such that \( \mathbb{E}_\mu[d(\mathbf{x}_0,\cdot)]<\infty \), where \( \mathbb{E}_\mu \) denotes the expectation operator associated with \( \mu \). Given two distributions \( \mu_1,\mu_2 \) over \( \Rset^{|\Vout|} \) with finite first moments w.r.t. some metric \( d \), the Kantorovich distance of these distributions w.r.t. \( d \) is the quantity
\[
\K_d(\mu_1 ,  \mu_2) = \sup_{f \in L^1_d(\Rset^{|\Vout|})} \Big|\mathbb{E}_{\mu_1}[f] - \mathbb{E}_{\mu_2}[f] \Big|, \]
where 
\[ L^1_d(\Rset^{|\Vout|}) = \Big\{f:\Rset^{|\Vout|}\rightarrow \mathbb{R} \,\Big|\, |f(x) - f(y)| \leq d(x,y) \text{ for all } x,y\in\Rset^{|\Vout|}\Big\} \]
is the set of all $1$-Lipschitz continuous functions defined over the metric space $(\Rset^{|\Vout|},d)$. When defined with respect to the discrete metric (which assigns distance $0$ to pairs of identical elements and distance $1$ otherwise), Kantorovich distance is equal to the Total Variation distance~\cite{villani2021topics}. We refer the reader to \cite[Section 3.3]{chatterjee2024equivalence} for a discussion of the motivation behind Kantorovich distance and its relationship with other measures of distributional similarity.

\smallskip\noindent {\bf Problem~2: Similarity refutation.}  Let $P^1$ and $P^2$ be two probabilistic programs with both $\pCFG(P^1)$ and $\pCFG(P^2)$ being a.s.~terminating and integrable pCFGs with the same set of output variables $\Vout$. Given (1)~a metric \( d \) over \( \mathbb{R}^{|\Vout|}  \) such that \( \mu^{\pCFG(P^1)} \) and \( \mu^{\pCFG(P^2)} \) have finite first moments w.r.t. \( d \), and (2)~some \( \varepsilon>0 \), our goal is to prove that \( P^1 \) and \( P^2 \) are \emph{not} \( \varepsilon \)\emph{-similar} w.r.t. the Kantorovich distance, i.e. that $\K_d(\mu^{\pCFG(P^1)} ,  \mu^{\pCFG(P^2)}) \geq \varepsilon$.

\subsection{Similarity Refutation in Programs with Conditioning}

We now present our new proof rule for similarity refutation in PPs with conditioning. Given a pair \( P^1,P^2 \) of a.s.~terminating PPs that satisfy the bounded cumulative weight property, our proof rule applies the weighted restarting procedure to produce the pCFGs \( \pCFG(\Prestarted^1) \) and \( \pCFG(\Prestarted^2) \) and then applies the proof rule of~\cite{chatterjee2024equivalence} for equivalence and similarity refutation in PPs without conditioning. The following theorem formalizes the instantiation of the proof rule of~\cite{chatterjee2024equivalence} to the pCFG pair \( \pCFG(\Prestarted^1) \) and \( \pCFG(\Prestarted^2) \). The proof of the theorem follows immediately by the correctness of the weighted restarting procedure (i.e.~Theorem~\ref{thm:weightedrestarting}) and the correctness of the proof rule in~\cite[Theorem~5.5]{chatterjee2024equivalence}.

\begin{theorem}[Equivalence and similarity refutation with conditioning]
	\label{thm:newproofrulesimilarity}
	Let \( P^1,P^2 \) be two a.s. terminating programs (possibly with conditioning) satisfying the bounded cumulative weight property. Let the initial states of \( \pCFG(\Prestarted^1) \) and \( \pCFG(\Prestarted^2) \) be \( (\locinit^1,\vecinit^1) \) and \( (\locinit^2,\vecinit^2) \), respectively.
	Assume that there exists a Borel-measurable function \( f \colon \Rset^{|V_\out|} \rightarrow \Rset \) and two state functions, \( U_f \) for \( \pCFG(\Prestarted^1) \) and \( L_f \) for \( \pCFG(\Prestarted^2) \) such that the following holds:
	\begin{compactitem}
		\item[i)] \( U_f \) is a UESM for \( f \) in \( \pCFG(\Prestarted^1) \) such that \( (\pCFG(\Prestarted^1), U_f, f) \) is OST-sound;
		\item[ii)] \( L_f \) is an LESM for \( f \) in \( \pCFG(\Prestarted^2) \) such that \( (\pCFG(\Prestarted^2), L_f, f) \) is OST-sound;
		\item[iii)] \( U_f(\locinit^1, \vecinit^1) + f((\vecinit^1)^\out) < L_f(\locinit^2, \vecinit^2) + f((\vecinit^2)^\out) \).
	\end{compactitem}
	Then \( P^1,P^2 \) are not output equivalent,
	
	Moreover, if \( f \) is \( 1 \)-Lipschitz continuous under a metric \( d \) over \( \Rset^{|V_\out|} \), then 
	$$\K_d\Big(\mu^{\pCFG(P^1)},\mu^{\pCFG(P^2)}\Big) \geq L_f(\locinit^2, \vecinit^2) + f((\vecinit^2)^\out) - U_f(\locinit^1, \vecinit^1) - f((\vecinit^1)^\out).$$
\end{theorem}

\section{Proof of Theorem \ref{thm:algo}}

\begin{proof}
	Suppose that the algorithm outputs ``Not output-equivalent'' and computes the ordered triple $(f, U_f, L_f)$. We show that these indeed define an instance of the proof rule in Theorem~\ref{thm:newproofrule}. Hence, it follows by Theorem~\ref{thm:newproofrule} that the algorithm is correct.
	
	Since the algorithm outputs ``Not output-equivalent'', we must have that $(f, U_f, L_f)$ provide a part of the solution to the system of constraints in Step~3. Thus, by  the correctness of the reduction to an LP instance in~\cite{ChatterjeeFG16,AsadiC0GM21,Wang0GCQS19,ZikelicCBR22}, it follows that they also provide a solution to the system of constraints collected in Step~2. But the constraints collected in Step~2 entail all the conditions imposed by the proof rule in Theorem~\ref{thm:newproofrule}. For all conditions except for OST-soundness, this follows immediately from the form of collected constraints. For OST-soundness, this follows by the correctness of the algorithm in~\cite{chatterjee2024equivalence} as we are using the same constraints to encode OST-soundness conditions. Hence, any solution to the system of constraints collected in Step~2 gives rise to $(f, U_f, L_f)$ which forms a correct instance of the proof rule in Theorem~\ref{thm:newproofrule}, as claimed.
	
	Second, we prove that the algorithm runs in polynomial time in the size of the pCFGs, for any fixed value of the polynomial degree parameter $D$. Suppose that the parameter $D$ is fixed. The number of constraints collected in Step~2 is linear in the number of pCFG locations, and each constraint is of size polynomial in the size of the pCFGs. Next, the reduction to an LP instance via Handelman's theorem gives rise to an LP instance which is of polynomial size, as shown in~\cite{chatterjee2024equivalence}, hence the size of the resulting LP remains polynomial in the size of the input pCFGs. Finally, since LP solving can be done in polynomial time, we conclude that the runtime of our algorithm is polynomial for a fixed value of $D$.
	
	Finally, for the semi-completeness claim, the algorithm of~\cite{chatterjee2024equivalence}, and hence our algorithm, in Step~2 collects a systemf of constraints which is satisfiable if and only if $(f,U_f,L_f)$ of a fixed polynomial template satisfy all the conditions of Theorem~\ref{thm:newproofrule} under OST-soundness condition (C2). Furthermore, the translation via Handelman's theorem in Step~3 is sound and complete for a fixed maximal polynomial degree $D$ whenever the satisfiability set of the left-hand-sides of all polynomial entailments is bounded, which is again ensured under the (C2) OST-soundness~\cite{AsadiC0GM21,ChatterjeeFG16}. Finally, LP solving in Step~4 is sound and complete. Hence, the semi-completeness claim follows. \hfill\qed
\end{proof}

\section{Proof of Theorem \ref{thm:c2algo}}

\begin{proof}
	We show that the value of every variable $v \in V$ of a PP $P$ that satisfies both (1) and (2) is bounded in every run of $P$. Each run of the $P$ is known to take at most $T$ steps. For each $v \in V$ and $0 \leq i \leq T$, let $M_i(v)$ be a non-negative real number such that the value of $v$ at the $i$-th step is bounded between $-M_i(v)$ and $M_i(v)$. For $i=0$, these bounds are specified by the initial variable valuation. Consider the transition $\tau$ being taken from $i$-th step to the $i+1$-st step. The update $\updates(\tau)$ has three possible cases for each variable:
	\begin{compactenum}
		\item If $\updates(\tau)(v) = \bot$, then $M_{i+1} = M_i$.
		\item If $\updates(\tau)(v) = g(x)$ for some  polynomial $g$ over $V$, then given that the current valuation of each variable is bounded by $M_i(\cdot)$, the value of $g$ will also be bounded. Hence, $M_{i+1}$ is the bound on $g(x)$.
		\item If $\updates(\tau)(v) = \delta$ for some probability distribution $\delta$, then $\delta$ is known to be bounded by the theorem assumptions, yielding a bound for $v$ in the next step that can be used as $M_{i+1}$.
	\end{compactenum}
	Since the program terminates in at most $T$ steps, we can obtain a global bound $M$ for all variables, such that the value of all variables always lies in $[-M,M]$ during a run of $P$. 
	
	Now, let $f, U_f$ and $L_f$ be as in the theorem statement. Since all three of these functions are polynomials in program variables and the program variables are known to be bounded in $[-M,M]$, it can be concluded that there is a bound $C$ such that the value of $f, U_f$ and $L_f$ stays in $[-C,C]$ during every run of $P$. This immediately implies that $|f+U_f|, |f+L_f| \leq 2C$ during every run of $P$, which means $P$ and the triple $(f,U_f, L_f)$ satisfy the (C2) OST-condition.
	\hfill\qed
\end{proof}

\end{document}